\documentclass[preprint]{elsarticle}
\usepackage{amssymb,amsmath,amscd,latexsym,epic,eepic}
\usepackage{times}
\usepackage{gastex}
\usepackage{paralist}
\usepackage{graphicx}
\usepackage{caption}
\usepackage{subcaption}

\newtheorem{theorem}{Theorem}
\newtheorem{corollary}{Corollary}
\newtheorem{lemma}{Lemma}
\newtheorem{example}{Example}
\newtheorem{remark}{Remark}
\newenvironment{proof}{Proof.}{}

\def\qed{\rule{0.4em}{1.4ex}}

 \newcommand{\set}[1]{\{#1\}}

\newcommand{\winval}[1]{\langle \! \langle #1 \rangle\! \rangle_{\mathit{val}} }
\newcommand{\va}{\winval{1}}
\newcommand{\vb}{\winval{2}}

\newcommand{\outcome}{\mathrm{Outcome}}
\newcommand{\Prb}{\mathrm{Pr}}

\newcommand{\trans}{\delta}
\newcommand{\dist}{{\cal D}}
\newcommand{\distr}{{\cal D}}

\newcommand{\calc}{{\cal C}}
\newcommand{\calg}{{\cal G}}

\newcommand{\IC}{{\sf PaC}}
\newcommand{\OC}{{\sf OsC}}
\newcommand{\PC}{{\sf CoC}}
\newcommand{\IT}{{\sf PaT}}
\newcommand{\OT}{{\sf OsT}}
\newcommand{\PT}{{\sf CoT}}

\newcommand{\I}{{\sf Pa}}
\newcommand{\Os}{{\sf Os}}
\renewcommand{\P}{{\sf Co}}
\newcommand{\C}{{\sf C}}
\newcommand{\T}{{\sf T}}

\newcommand{\slopefrac}[2]{\leavevmode\kern.1em
  \raise .5ex\hbox{\the\scriptfont0 #1}\kern-.1em
  /\kern-.15em\lower .25ex\hbox{\the\scriptfont0 #2}}
\newcommand{\half}{\slopefrac{1}{2}}

\newcommand{\straa}{\sigma} \newcommand{\Straa}{\Sigma}
\newcommand{\strab}{\pi} \newcommand{\Strab}{\Pi}

\newcommand{\nats}{\mathbb{N}} \newcommand{\Nats}{\mathbb{N}}
\newcommand{\reals}{\mathbb{R}} 
\newcommand{\nat}{\mathbb N}

\newcommand{\real}{{\mathbb R}}

\newcommand{\Inf}{\mathrm{Inf}} 
\newcommand{\ov}[1]{#1'}

\newcommand{\Pref}{{\sf Prefs}}
\newcommand{\Play}{{\sf Plays}}

\newcommand{\POMDP}{{\sf POMDP}}
\newcommand{\MDP}{{\sf MDP}}
\newcommand{\tuple}[1]{\langle #1 \rangle}

\newcommand{\Buchi}{\mathsf{Buchi}}
\newcommand{\coBuchi}{\mathsf{coBuchi}}
\newcommand{\Parity}{\mathsf{Parity}}
\newcommand{\Outcome}{{\mathsf{Outcome}}}
\newcommand{\target}{{\cal T}}
\newcommand{\Obs}{{\cal{O}}}
\newcommand{\obs}{\mathsf{obs}}

\newcommand{\Succ}{\mathsf{Succ}}
\newcommand{\sink}{\mathsf{sink}}

\begin{document}

\title{Randomness for Free\tnoteref{t1,t2}}

\tnotetext[t1]{A preliminary version of this paper appeared in the 
\emph{Proceedings of the 35th International Symposium on Mathematical 
Foundations of Computer Science} (MFCS), 
Lecture Notes in Computer Science 6281, Springer, 2010, pp. 246-257.}

\tnotetext[t2]{
This research was partly supported by Austrian Science Fund (FWF) 
Grant No P 23499- N23, FWF NFN Grant No S11407-N23 and  S11402-N23 (RiSE), 
ERC Start grant (279307: Graph Games), Microsoft faculty fellows award, 
the ERC Advanced Grant QUAREM (267989: Quantitative Reactive Modeling), 
European project Cassting (FP7-601148), European project COMBEST, and the European
Network of Excellence ArtistDesign.
}

\author[IST]{Krishnendu Chatterjee}
\ead{Krishnendu.Chatterjee@ist.ac.at}

\author[LSV]{Laurent~Doyen\corref{cor}}
\ead{doyen@lsv.ens-cachan.fr}

\author[LaBRI]{Hugo Gimbert}
\ead{hugo.gimbert@labri.fr}

\author[IST]{Thomas~A.~Henzinger}
\ead{tah@ist.ac.at}

\cortext[cor]{Corresponding author: Laurent Doyen, LSV, CNRS UMR 8643 \& ENS Cachan, 61 avenue du Pr\'esident Wilson, 94235  Cachan Cedex, France.}

\address[IST]{IST Austria (Institute of Science and Technology Austria)}

\address[LSV]{LSV, ENS Cachan \& CNRS, France}

\address[LaBRI]{LaBRI \& CNRS, Bordeaux, France}

\begin{abstract}
We consider two-player zero-sum games on finite-state graphs. 
These games can be classified on the basis of the information of the players 
and on the mode of interaction between them. 
On the basis of information the classification is as follows: 
(a)~partial-observation (both players have partial view of the game); 
(b)~one-sided complete-observation (one player has complete observation);
and (c)~complete-observation (both players have complete view of the game).
On the basis of mode of interaction we have the following classification: 
(a)~concurrent (players interact simultaneously); and (b)~turn-based 
(players interact in turn).
The two sources of randomness in these games are randomness in the transition 
function and randomness in the strategies.
In general, randomized strategies are more powerful than deterministic 
strategies, and probabilistic transitions give more general classes of games.
We present a complete characterization for the classes of games 
where randomness is not helpful in: (a)~the transition function 
(probabilistic transitions can be simulated by deterministic transitions); and
(b)~strategies (pure strategies are as powerful as randomized strategies).
As a consequence of our characterization we obtain new undecidability 
results for these games.
\end{abstract}

\maketitle

\section{Introduction}
\noindent{\bf Games on graphs.} 
Games played on graphs provide the mathematical framework to analyze 
several important problems in computer science as well as mathematics.  
In particular, when the vertices and edges of a graph represent the states 
and transitions of a reactive system, then the synthesis problem
(Church's problem) asks for the construction of a winning strategy in a 
game played on a graph 
\cite{BuchiLandweber69,RamadgeWonham87,PnueliRosner89,McNaughton93}.
Game-theoretic formulations have also proved useful for the 
verification~\cite{AHK02}, refinement~\cite{FairSimulation}, and compatibility 
checking \cite{InterfaceTheories} of reactive systems.    
Games played on graphs are dynamic games that proceed for an infinite number 
of rounds. 
In each round, the players choose moves; the moves, together with
the current state, determine the successor state.
An outcome of the game, called a {\em play}, consists of the infinite
sequence of states that are visited.

\smallskip\noindent{\bf Strategies and objectives.}
A strategy for a player is a recipe that describes how the player chooses a 
move to extend a play.
Strategies can be classified as follows:
(a)~{\em pure} strategies, which always deterministically choose a  
move to extend the play, and (b)~{\em randomized\/} strategies, which may choose at a state a probability
distribution over the available moves.
Objectives are generally Borel-measurable sets~\cite{Kechris}: %\cite{Mar98}: 
the objective  for a player is a Borel set $B$ in
the Cantor topology on $S^\omega$ (where $S$ is the set of states),
and the player satisfies the objective if the outcome of the game is a 
member of~$B$.
In verification, objectives are usually  {\em $\omega$-regular languages}.
The $\omega$-regular languages generalize the classical regular languages
to infinite strings; they occur in the low levels of the Borel hierarchy 
(they lie in $\Sigma_3\cap\Pi_3$) 
and they form a robust and expressive language for determining payoffs
for commonly used specifications. 
%%The simplest $\omega$-regular objectives correspond to ``safety'' 
%%(the closed sets in the topology of~$S^\omega$)
%%and ``reachability'' (the open sets).

\smallskip\noindent{\bf Classification of games.}
Games played on graphs can be classified according to the 
knowledge of the players about the state of the game, and the way of 
choosing moves.
Accordingly, there are 
(a)~\emph{partial-observation} games, where each player 
only has a partial or incomplete view about the state and the moves 
of the other player; 
(b)~\emph{one-sided complete-observation} games, where one player has 
partial knowledge and the other player has complete knowledge about 
the state and moves of the other player; and 
(c)~\emph{complete-observation} games, where each player has complete 
knowledge of the game.
According to the way of choosing moves, the games on graphs can be  
classified into {\em turn-based\/} and {\em concurrent\/} games. 
In turn-based games, in any given round only one player can choose
among multiple moves; effectively, the set of states
can be partitioned into the states where it is player~1's turn to
play, and the states where it is player~2's turn. 
In concurrent games, both players may have multiple moves available 
at each state, and the players choose their moves simultaneously and
independently.

\smallskip\noindent{\bf Sources of randomness.} There are two sources
of randomness in these games. First is the randomness in the transition 
function: given a current state and moves of the players, the transition 
function defines a probability distribution over the successor states. 
The second source of randomness is the randomness in strategies (when 
the players play randomized strategies). In this work we study when 
randomness can be obtained for \emph{free}; i.e., we study in which classes
of games the probabilistic transitions can be simulated by 
deterministic transitions and the classes of games where pure 
strategies are as powerful as randomized strategies.

\smallskip\noindent{\bf Motivation.} 
The motivation to study this problem is as follows:
(a)~if for a class of games it can be shown that randomness is for free in the
transition function, then all future works related to analysis of computational
complexity, strategy complexity, and algorithmic solutions can focus on the 
simpler class with deterministic transitions (the randomness in transition function
may be essential for modeling appropriate stochastic reactive systems, 
but the analysis can focus on the deterministic subclass); 
(b)~if for a class of games it can be shown that randomness is for free in 
strategies, then all future works related to correctness results 
can focus on the simpler class of pure strategies, and the 
results would follow for the more general class of randomized strategies; 
and (c)~the characterization of randomness for free will allow hardness results
obtained for the more general class of games (such as games with randomness in 
the transition function) to be carried over to simpler class of games (such as games 
with deterministic transitions).

\smallskip\noindent{\bf Contribution.} The contributions of this paper are as follows:
\begin{enumerate} 
\item {\em Randomness for free in the transition function.}
We show that randomness in the transition function can be obtained 
for free for complete-observation concurrent games (and any class that subsumes
complete-observation concurrent games) and for one-sided complete-observation 
turn-based games (and any class that subsumes this class). 
The reduction is polynomial for complete-observation concurrent games, and 
exponential for one-sided complete-observation turn-based games.
It is known that for complete-observation turn-based games, 
a probabilistic transition function cannot be simulated by a deterministic 
transition function (see discussion in Section~\ref{sec:remarks} for details), and 
thus we present a complete characterization when randomness can be obtained 
for free in the transition function.

\item {\em Randomness for free in the strategies.}
We show that randomness in strategies is free for complete-observation  
turn-based games, and for $1$-player partial-observation games (POMDPs). 
For all other classes of games randomized strategies are more powerful than 
pure strategies.
It follows from a result of Martin~\cite{Mar98} that for $1$-player 
complete-observation games with probabilistic transitions (MDPs) pure 
strategies are as powerful as randomized strategies. 
We present a generalization of this result to the case of %one-player partial-observation  games with probabilistic transitions (POMDPs).
POMDPs.
Our proof is totally different from Martin's proof and 
based on a new derandomization technique of randomized strategies.
%%We extend this result to the case of one-player 
%%partial-observation  games with probabilistic transition (POMDPs).
%%This requires a proof totally different from Martin's proof,
%%based on a derandomization technique of randomized strategies.
  
\item {\em Concurrency for free in games.} 
We show that concurrency is obtained for free with partial-observation,
both for one-sided complete-observation games as well as for general 
partial-observation games (see Section~\ref{sec:concurrency4free}).
It follows that for partial-observation games, future research 
can focus on the simpler model of turn-based games, and
concurrency does not add anything in the presence of partial observation.

\item {\em New undecidability results.}
As a consequence of our characterization of randomness for free,
we obtain new undecidability results.
In particular, using our results and results of 
Baier et al.~\cite{BBG08} we show for one-sided complete-observation 
deterministic games, the problems of almost-sure winning for 
coB\"uchi objectives and positive winning for B\"uchi objectives 
are undecidable.
Thus we obtain the first undecidability result for qualitative analysis 
(almost-sure and positive winning) of one-sided complete-observation 
deterministic games with $\omega$-regular objectives.
\end{enumerate}

\smallskip\noindent{\bf Applications of our results.} 
While we already show that our results allow us to obtain new undecidability 
results, they have also been used to simplify proofs and analysis of 
POMDPs and partial-observation games~\cite{CCHRS11,CC14,CCT13,CD14,GO14} 
(e.g.~\cite[Lemma~21]{CC14} and~\cite[Claim~2. Lemma~5.1]{CD14}) 
as well as extended to other settings such as 
probabilistic automata~\cite{GS14}.

\section{Definitions}
In this section we present the definition of concurrent games of 
partial information and their subclasses, and %%related 
notions of strategies and objectives. Our model of game is 
equivalent to the model of stochastic games with signals~\cite{MSZ94,DBLP:conf/lics/BertrandGG09}
(in stochastic games with signals, the players receive signals which 
represent information about the game, which in our model is represented as 
observations).
A \emph{probability distribution} on a finite set $A$ is a function
$\kappa: A \to [0,1]$ such that $\sum_{a \in A} \kappa(a) = 1$. 
%The \emph{support} of $\kappa$ is the set $\Supp(\kappa) = \{a \in A \mid \kappa(a) > 0\}$.
We denote by $\dist(A)$ the set of probability distributions on $A$.

\smallskip\noindent{\bf Concurrent games of partial observation.}
A \emph{concurrent game of partial observation} 
(or simply a \emph{game}) is a tuple 
$G=\tuple{S,A_1,A_2,\trans,\Obs_1,\Obs_2}$ with the following components: 
\begin{enumerate}
\item \emph{(State space).} $S$ is a finite set of states; 
\item \emph{(Actions).} $A_i$ ($i=1,2$) is a finite set of actions for 
player~$i$; 
\item \emph{(Probabilistic transition function).}
$\trans : S \times A_1 \times A_2 \to \dist(S)$ is a concurrent probabilistic 
transition function that given a current state $s$, 
actions~$a_1$ and~$a_2$ for both players gives the transition probability 
$\trans(s,a_1,a_2)(s')$ to the next state~$s'$; 
for the sake of effectiveness, we assume that all probabilities in the transition function are rational;
\item \emph{(Observations).} $\Obs_i \subseteq 2^S$ ($i=1,2$) is a finite set 
of observations for player~$i$ that partition the state space~$S$.
These partitions uniquely define functions $\obs_i: S \to \Obs_i$ ($i=1,2$) 
that map each state to its observation (for player~$i$) such that $s \in \obs_i(s)$ for all $s \in S$.
\end{enumerate}

We sometimes relax the assumption that games 
have a finite state space, and we allow the set $S$ of states to be \emph{countable}.
This is useful in the context of game solving, where we get a countable state space
after fixing an arbitrary strategy for one of the players in a game.
In our results we explicitly mention when we consider countable state space
and when we consider finite state space.

%%\mynote{L: should we allow transition functions to be non-total, and say that if the players choose a 
%%pair of actions $(a_1,a_2)$ such that $\trans(s,a,b)$ is not defined, then both player loose; and say that 
%%we can/should adapt the definition of play accordingly (plays can be finite then - note that finite plays 
%%never belong to the objective) but we do not do so for the sake of clarity.
%%Krish answer: I think we can keep games as total: it removes lot of unnecessary technical details: 
%%otherwise we have finite vs infinite plays; non-zero-sum objectives. More or less it is agreed that 
%%total assumption is just technical convenience.
%%}

\noindent{\bf Special cases.} We consider the following special cases of 
partial-observation concurrent games, obtained either by restrictions 
in the observations, the mode of selection of moves,  
the type of transition function, or the number of players:
\begin{itemize}
\item \emph{(Observation restriction).}
The games with \emph{one-sided complete-observation} are the special case of 
games where $\Obs_1 = \{ \{s\} \mid s \in S \}$ (i.e., player~1 has 
complete observation) or $\Obs_2 = \{ \{s\} \mid s \in S \}$ (player~2 has 
complete observation).
The \emph{games of complete-observation} are the special case of games where 
$\Obs_1 = \Obs_2 = \{ \{s\} \mid s \in S \}$, i.e., every state is visible to 
each player and hence both players have complete observation. 
If a player has complete observation we omit the corresponding observation 
sets from the description of the game.

\item \emph{(Mode of interaction restriction).} 
A \emph{turn-based state} is a state $s$ such that either
$(i)$ $\trans(s,a,b) = \trans(s,a,b')$ for all $a \in A_1$
and all $b,b' \in A_2$ (i.e, the action of player~1 determines the transition function
and hence it can be interpreted as player~1's turn to play), we refer to~$s$ as 
a player-1 state, and we use the notation $\trans(s,a,-)$; 
or $(ii)$ $\trans(s,a,b) = \trans(s,a',b)$ for all $a,a' \in A_1$ and all $b \in A_2$, 
we refer to~$s$ as a player-2 state, and we use the notation $\trans(s,-,b)$.
A state $s$ which is both a player-1 state and a player-2 state is called a \emph{probabilistic state}
(i.e., the transition function is independent of the actions of the players). 
We write $\trans(s,-,-)$ to denote the transition function in~$s$.
The \emph{turn-based games} are the special case of games where 
all states are turn-based.  

\item \emph{(Transition function restriction).} 
The \emph{deterministic games} are the special case of games where 
for all states $s \in S$ and actions $a \in A_1$ and $b \in A_2$,
there exists a state $s' \in S$ such that $\trans(s,a,b)(s') = 1$. 
We refer to such states~$s$ as deterministic states.
For deterministic games, it is often convenient to assume 
that $\trans: S \times A_1 \times A_2 \to S$. 

\item \emph{(Player restriction).} 
The \emph{1\half-player games}, also called \emph{partially observable 
Markov decision processes} (or \POMDP{}s),
are the special case of games where the action set~$A_1$ or~$A_2$ is a singleton. 
Note that 1\half-player games are turn-based. Games without
player restriction are sometimes called 2\half-player games.
\end{itemize}

The 1\half-player games of complete-observation are Markov decision processes 
(or \MDP{}s), and \MDP{}s with all states deterministic 
%1\half-player deterministic games 
can be viewed as graphs (and are often called $1$-player games).

\smallskip\noindent{\em Classes of game graphs.} We use the following 
abbreviations (\tablename~\ref{tab:abbreviations-a}): we write \I\ for partial-observation, \Os\ for one-sided 
complete-observation, \P\ for complete-observation, \C\ for concurrent, and \T\ 
for turn-based. 
For example, \PC\ will denote complete-observation concurrent games, and \OT\ 
will denote one-sided complete-observation turn-based games. 
For $\calc \in \set{\I,\Os,\P} \times \set{\C,\T}$, %\set{\IC,\OC,\PC,\IT,\OT,\PT}$ 
we denote by $\calg_\calc$ the set of all $\calc$ games.
Note the following strict inclusions (see also \figurename~\ref{figure-overview}): partial observation (\I) is more general
than one-sided complete-observation (\Os) and \Os\ is more general than complete-observation 
(\P), and concurrent (\C) is more general than turn-based (\T).
We will denote by $\calg_D$ the set of all games with deterministic transition
function. The results we establish in this article are summarized in \figurename~\ref{figure-overview3}.

\smallskip\noindent{\em Plays.}
In concurrent games of partial observation, in each turn, player~$1$ chooses an action $a \in A_1$, 
player~$2$ chooses an action $b \in A_2$, and the successor of the current
state~$s$ is chosen according to the probabilistic transition function $\trans(s,a,b)$.
%
% A \emph{play} in $G$ is an infinite sequence of states 
% $\rho=s_0 s_1 \ldots$ such that for all $i \geq 0$,
% there exists $a_i \in A_1$ and $b_i \in A_2$ with
% $\trans(s_i,a_i,b_i,s_{i+1}) > 0$.
% The \emph{prefix up to $s_n$} of the play $\rho$ is denoted by $\rho(n)$, 
% its \emph{length} is $\abs{\rho(n)} = n+1$ and its \emph{last element} is 
% $\Last(\rho(n)) = s_n$. The set of plays in $G$ is denoted $\Play(G)$,
% and the set of corresponding finite prefixes is denoted $\Pref(G)$.
% The \emph{observation sequence} of $\rho$ for player~$i$ ($i=1,2$) is the unique 
% infinite sequence $\obs_i(\rho)=o_0 o_1 \ldots \in O_i^{\omega}$ 
% such that $s_j \in o_j$ %, $c_j = a_j$ if $i=1$, and $c_j = b_j$ if $i=2$
% for all $j \geq 0$. 
A \emph{play} in a game~$G$ is an infinite sequence $\rho=s_0 \,a_0b_0\, s_1 \,a_1b_1\, s_2 \ldots$ 
such that \mbox{$\trans(s_i,a_i,b_i,s_{i+1}) > 0$} for all $i \geq 0$.
The \emph{prefix up to $s_n$} of the play $\rho$ is denoted by $\rho(n)$. %, 
%its \emph{length} is $\abs{\rho(n)} = n+1$ and its \emph{last element} is $\Last(\rho(n)) = s_n$. 
The set of plays in $G$ is denoted $\Play(G)$,
and the set of corresponding finite prefixes (or histories) is denoted $\Pref(G)$.
The \emph{observation sequence} of $\rho$ for player~$i$ ($i=1,2$) is the unique 
infinite sequence $\obs_i(\rho)=o_0 \,c_0\, o_1 \,c_1\, o_2\ldots$ 
such that $s_j \in o_j \in O_i$, and $c_j = a_j$ if $i=1$, and $c_j = b_j$ if $i=2$
for all $j \geq 0$.

\begin{table}[t!]
    \centering
    \begin{subfigure}[t]{0.5\textwidth}
	\centering
 	\begin{tabular}{c|c}
	  $\I$  & partial observation \\
	  $\Os${\large \strut} & one-sided complete observation \\
 	  $\P${\large \strut}  & complete observation \\
	\hline
	  $\C${\large \strut}  & concurrent \\
	  $\T$  & turn-based \\
	\hline
	  $D${\large \strut}   & deterministic transition function
	\end{tabular}
        \caption{Classes of games \label{tab:abbreviations-a}}
    \end{subfigure}%
~
    \begin{subfigure}[t]{0.5\textwidth}
        \centering
 	\begin{tabular}{c|c}
	  $\Straa_G$    & all player-$1$ strategies  \\
	  $\Straa_G^O${\large \strut}  & observation-based pl.-$1$ strategies\\
 	  $\Straa_G^P${\large \strut}  & pure player-$1$ strategies \\
	\hline
	  $\Strab_G$    & all player-$2$ strategies \\
	  $\Strab_G^O${\large \strut}  & observation-based pl.-$2$ strategies \\
	  $\Strab_G^P${\large \strut}  & pure player-$2$ strategies \\
	\end{tabular}
        \caption{Classes of strategies in game $G$ \label{tab:abbreviations-b}}
    \end{subfigure}
    \caption{Abbreviations. \label{tab:abbreviations}}
\end{table}

\smallskip\noindent{\em Strategies.}
A \emph{pure strategy} in a game $G$ for player~$1$ is a function $\straa:\Pref(G) \to A_1$. 
A \emph{randomized strategy} in $G$ for player~$1$ is a function $\straa:\Pref(G) \to \dist(A_1)$. 
A (pure or randomized) strategy $\straa$ for player~$1$ is 
\emph{observation-based} if for all prefixes $\rho,\rho' \in \Pref(G)$, 
if $\obs_1(\rho)=\obs_1(\rho')$, then $\straa(\rho)=\straa(\rho')$. 
%In the sequel, we are interested in the existence of observation-based strategies for player~$1$.
We omit analogous definitions of strategies for player~$2$.
We denote by $\Straa_G$, $\Straa_G^O$, $\Straa_G^P$, $\Strab_G$, $\Strab_G^O$ and $\Strab_G^P$ 
the set of all player-$1$ strategies in $G$, the set of all observation-based player-$1$ strategies, 
the set of all pure player-$1$ strategies, the set of all player-$2$ strategies in $G$, the set of all 
observation-based player-$2$ strategies, and the set of all pure player-$2$ strategies, 
respectively (\tablename~\ref{tab:abbreviations-b}).
Note that if player~$1$ has complete observation, then $\Straa_G^O=\Straa_G$.
%

%%%% Outcome not needed
\begin{comment}
The \emph{outcome} of two randomized strategies $\straa$ (for
player~$1$) and $\strab$ (for player~$2$) in $G$ is the set of plays
$\rho=s_0 s_1 \ldots \in \Play(G)$ where for all $i \geq 0$, there exist
actions $a_i \in A_1$ and $b_i \in A_2$ such that
$\straa(\rho(i))(a_i) > 0$, $\strab(\rho(i))(b_i)>0$, and $\trans(s_i,a_i,b_i)(s_{i+1}) > 0$.
This set is denoted $\outcome(G,\straa,\strab)$.
The outcome of two pure strategies is defined analogously by viewing  
pure strategies $\straa$ as randomized strategies that play the action
chosen by $\straa$ with probability one.
The \emph{outcome set} of the pure (resp.\ randomized)
strategy $\straa$ for player~$1$ in $G$ is the set
$\Outcome_1(G,\straa)$ of plays $\rho$ such that there exists a
pure (resp.\ randomized) strategy $\strab$ for
player~$2$ with $\rho\in\outcome(G,\straa,\strab)$.
The outcome set $\Outcome_2(G,\strab)$ for player~2 is defined symmetrically.
\end{comment}

\smallskip\noindent{\em Objectives.}
An \emph{objective} for player~$1$ in $G$ is a set $\varphi \subseteq S^\omega$ of 
infinite sequences of states. 
A play $\rho = s_0 \,a_0b_0\, s_1 \,a_1b_1\, s_2 \ldots \in \Play(G)$ 
\emph{satisfies} the objective $\varphi$, denoted $\rho \models \varphi$, if 
$s_0 s_1 s_2 \ldots \in \varphi$.
A Borel objective is a Borel-measurable set in the Cantor topology on $S^\omega$~\cite{Kechris}.
%Objectives are generally Borel measurable: 
%a Borel objective is a Borel set in the Cantor topology on $S^\omega$~\cite{Kechris}.
We specifically consider $\omega$-regular objectives 
specified as parity objectives (a canonical form to express all $\omega$-regular 
objectives~\cite{Thomas97}).
For a sequence $\bar{s} = s_0 s_1 s_2 \ldots$ we denote by $\Inf(\bar{s})$ 
the set of states that occur infinitely often in $\bar{s}$, that is, 
$\Inf(\bar{s})=\{ s \in S \mid s_j=s \text{ for infinitely many } j \text{'s} \}$.
For $d \in \Nats$, let $p:S \to \{0,1,\ldots,d\}$ be a 
\emph{priority function}, 
which maps each state to a nonnegative integer priority.
The \emph{parity} objective $\Parity(p)$ requires that the minimum priority 
that occurs infinitely often be even.
Formally, $\Parity(p)=\{\bar{s} \in S^\omega \mid \min\set{ p(s) \mid s \in \Inf(\bar{s})} \text{ is even} \}$.
The B\"uchi and coB\"uchi objectives are the special cases 
of parity objectives with two priorities, for $p: S \to \set{0,1}$ and 
$p: S \to \set{1,2}$ respectively.
We say that an objective $\varphi$ is \emph{visible} for player~$i$ if for all $\rho,\rho' \in \Play(G)$,
if $\rho \models \varphi$ and $\obs_i(\rho)=\obs_i(\rho')$, then $\rho'\models \varphi$.
For example if the priority function maps observations to priorities 
(i.e., $p: \Obs_i \to \set{0,1,\ldots,d}$), then the parity objective is 
visible for player~$i$.

\smallskip\noindent{\em Almost-sure winning, positive winning, and value function.}
%%A strategy $\straa$ for player~$1$ in $G$ is \emph{sure winning} 
%%for an objective $\varphi$ if $\rho \models \varphi$ for all $\rho \in \Outcome_1(G,\straa)$.
An \emph{event} is a measurable subset of $S^{\omega}$, and 
given strategies $\straa$ and $\strab$ for the two players, 
the probabilities of events are uniquely defined~\cite{Var85}. 
For a Borel objective~$\varphi$, we denote by $\Prb_{s}^{\straa,\strab}(\varphi)$ 
the probability that $\varphi$ is satisfied by the play obtained from the starting state $s$ when the 
strategies $\straa$ and $\strab$ are used.
Given a game structure $G$ and a state $s$, an observation-based  
strategy $\straa$ for player~$1$ is 
\emph{almost-sure winning} 
(resp., \emph{positive winning}) for the objective $\varphi$ from $s$ if 
for all observation-based randomized strategies $\strab$ for 
player~$2$, we have $\Prb_{s}^{\straa,\strab}(\varphi)=1$ 
(resp.,  $\Prb_{s}^{\straa,\strab}(\varphi)>0$).

The \emph{value function} $\va^G(\varphi):S \to \reals$ for player~1 and  
objective~$\varphi$ assigns to every state of $G$ the maximal probability 
with which player~1 can guarantee the satisfaction of $\varphi$ 
with an observation-based strategy, against all observation-based 
strategies for player~2. Formally we define
\[
\va^G(\varphi)(s)= \sup_{\straa \in \Straa_G^O} \inf_{\strab \in \Strab_G^O}
\Prb_{s}^{\straa,\strab}(\varphi).
\]

The value of an observation-based strategy $\straa$ for player~1 and  
objective~$\varphi$ in state $s$ is 
$val_1^{\straa}(\varphi)(s) = \inf_{\strab \in \Strab_G^O} \Prb_{s}^{\straa,\strab}(\varphi)$.
Analogously for player~2, define 
$\vb^G(\varphi)(s)= \inf_{\strab \in \Strab_G^O} \sup_{\straa \in \Straa_G^O} 
\Prb_{s}^{\straa,\strab}(\varphi)$
and $val_2^{\strab}(\varphi)(s) = \sup_{\straa \in \Strab_G^O} \Prb_{s}^{\straa,\strab}(\varphi)$.
For $\varepsilon \geq 0$, an observation-based strategy $\straa$ is \emph{$\varepsilon$-optimal}
for $\varphi$ from $s$ if $val_1^{\straa}(\varphi)(s) \geq \va^G(\varphi)(s) - \varepsilon$.
An \emph{optimal} strategy is a $0$-optimal strategy.

\begin{figure}[t]
   \begin{center}
	\hrule
      %\unitlength=.6mm
%\def\fsize{\scriptsize}
\def\fsize{\normalsize}

\begin{picture}(115,65)(0,0)
%\put(0,0){\framebox(115,65){}}

{\fsize
\gasset{rdist=1,Nadjust=n}

\node[Nmarks=i,NLangle=0.0](q0)(10,29){{\fsize $s_1$}}
\node[Nmarks=n](q1)(40,44){$s_2$}
\node[Nmarks=n](q2)(40,14){$s'_2$}
\node[Nmarks=n](q3)(75,44){$s_3$}
\node[Nmarks=n](q4)(75,14){$s'_3$}
\node[Nmarks=r](q5)(105,14){$s_4$}

%\drawloop[ELside=l,loopCW=y, loopdiam=6](q0){$b$}
%\drawloop[ELside=l,loopCW=y, loopdiam=6](q1){$b$}
%\drawloop[ELside=l,loopCW=y, loopdiam=6](q2){$b$}
%\drawloop[ELside=l,loopCW=y, loopdiam=6](q3){$b$}
%\drawloop[ELside=l,loopCW=y, loopdiam=6](q4){$b$}
%\drawloop[ELside=l,loopCW=y, loopangle=0, loopdiam=6](q6){$a,b$}
\drawloop[ELside=l,loopCW=y, loopangle=90, loopdiam=6](q5){$(-,-)$}

\drawedge[ELpos=50,ELside=l](q0,q1){$(-,b_1)$}
\drawedge[ELpos=50,ELside=r](q0,q2){$(-,b_2)$}
\drawedge(q1,q3){$(a_1,-)$}
\drawedge[ELside=r](q2,q4){$(a_1,-)$}
\drawedge[ELpos=35, ELside=l, ELdist=.4](q1,q4){$(a_2,-)$}
\drawedge[ELpos=35, ELside=r, ELdist=.4](q2,q3){$(a_2,-)$}
%\drawedge[ELpos=55, ELside=r, ELdist=1, curvedepth=-18](q3,q0){$a,b$}
\drawbpedge[ELpos=78, ELside=r, ELdist=1, syo=3, exo=-1](q3,165,20,q0,80,50){$(-,-)$}
\drawedge[ELpos=50](q4,q5){$(-,-)$}

\gasset{Nmr=0,Nframe=y,Nadjust=n,dash={1.5}0,AHnb=0}

\node[Nw=16,Nh=13](c1)(10,29){}
\node[Nw=16,Nh=43](c2)(40,29){}
\node[Nw=16,Nh=43](c3)(75,29){}
\node[Nw=16,Nh=13](c4)(105,14){}

\gasset{Nframe=n}

\node(o1)(10,3){$o_1$}
\node(o2)(40,3){$o_2$}
\node(o3)(75,3){$o_3$}
\node(o4)(105,3){$o_4$}

%\drawedge[dash={1}0](n3bis,nkbis){$0,1$}
}
\end{picture}
	\hrule
  \vspace{6pt}
  \caption{A game with one-sided complete observation (Example~\ref{ex:example-one}).  \label{figure-example1}}
   \end{center}
\end{figure}
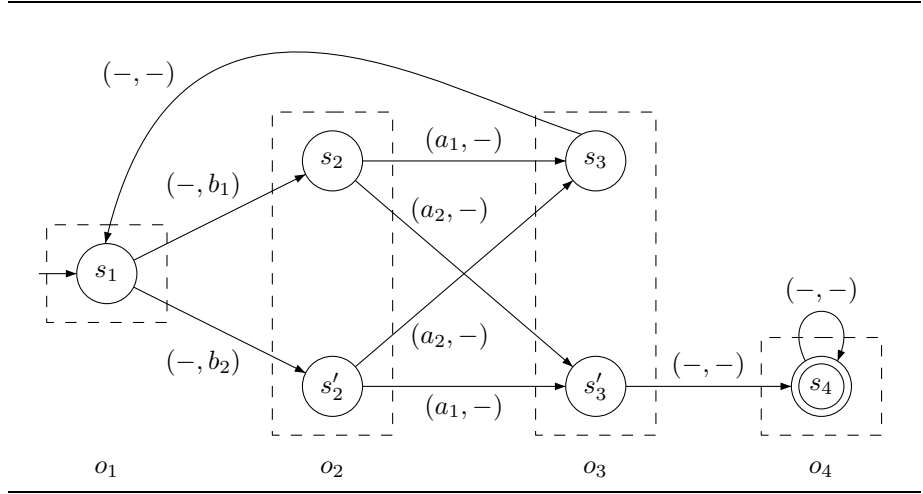

\begin{example}[\cite{CDHR07}]\label{ex:example-one}
Consider the game with one-sided complete observation (player~$2$ has complete information)
shown in \figurename~\ref{figure-example1}. 
Consider the B\"uchi objective defined by the state $s_4$ (i.e., state $s_4$ has priority $0$
and other states have priority $1$). Because player~$1$ has partial observation (given
by the partition $\Obs_1  = \{ \{s_1\},\{s_2,s'_2\},\{s_3,s'_3\},\{s_4\} \}$), she
cannot distinguish between $s_2$ and $s'_2$ and therefore has to play the same actions 
with same probabilities in $s_2$ and $s'_2$ (while it would be easy to win by playing
$a_2$ in $s_2$ and $a_1$ in $s'_2$, this is not possible). In fact, player~$1$
cannot win using a pure observation-based strategy. However, playing $a_1$ and $a_2$ uniformly at 
random in all states is almost-sure winning. 
Every time the game visits observation $o_2$, for any strategy of player~$2$,
the game visits $s_3$ and $s'_3$ with probability $\frac{1}{2}$, and hence
also reaches $s_4$ with probability $\frac{1}{2}$. 
It follows that against all player-$2$ strategies the play 
eventually reaches $s_4$ with probability~1, and then stays there.
\end{example}

\begin{theorem}[\cite{Mar98}]\label{thrm:martin}
Let $G$ be a \PT\ stochastic game (with countable state space $S$) 
with initial state $s$ and an %%%Borel-measurable 
objective $\varphi\subseteq S^\omega$.
Then the following equalities hold:
%%\begin{equation}\label{eq:martin}
$\va^G(\varphi)(s) =  
\vb^G(\varphi)(s) = 
\sup_{\straa \in \Straa_G^O \cap \Straa_G^P} \inf_{\strab \in \Strab_G^O} \Prb_{s}^{\straa,\strab}(\varphi)$.
%%\enspace.
%%\end{equation}
\end{theorem}

\smallskip\noindent{\bf Discussion of Theorem~\ref{thrm:martin}.}
Theorem~\ref{thrm:martin} can be derived as a consequence of Martin's proof of 
determinacy of Blackwell games~\cite{Mar98}: the result states that for \PT\ 
stochastic games pure strategies can achieve the same value as randomized
strategies, and as a special case, the result also holds for 
\MDP s (for a detailed discussion how to obtain the result from~\cite{Mar98} 
see~\cite[Lemma~10]{CMJ04}).
Note that Martin's determinacy result of $\va^G(\varphi)(s) =  \vb^G(\varphi)(s)$ also
holds for \PC\ stochastic games (complete-observation concurrent stochastic games), 
but the equality with $\sup_{\straa \in \Straa_G^O \cap \Straa_G^P} \inf_{\strab \in \Strab_G^O} \Prb_{s}^{\straa,\strab}(\varphi)$
(which implies existence of pure $\epsilon$-optimal strategies for $\epsilon>0$) 
only holds for \PT\ stochastic games.

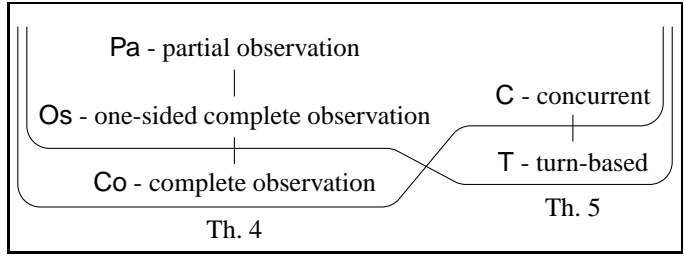
\begin{figure}[!tb]
   \begin{center}
      \unitlength=.6mm

\begin{picture}(149,55)(0,0)
\put(0,0){\framebox(149,55){}}

\gasset{Nw=9,Nh=9,Nmr=4.5,rdist=1,Nadjust=n}

\node[Nframe=n](q3)(50,45){\I\/ - partial observation}
\node[Nframe=n](q2)(50,30){\Os\/ - one-sided complete observation}
\node[Nframe=n](q1)(50,15){\P\/ - complete observation}

\node[Nframe=n](p2)(125,35){\C\/ - concurrent}
\node[Nframe=n](p1)(125,20){\T\/ - turn-based}

\drawedge[AHnb=0](q1,q2){}
\drawedge[AHnb=0](q3,q2){}
\drawedge[AHnb=0](p1,p2){}

%\drawline[AHnb=0](3,40)(5,20)(85,20)(95,10)(145,10)(148,30)

%\drawcurve[AHnb=0](3,50)(5,12)(10,11)(80,11)(87,12)(100,24)(102,25)(132,26)(139,27)(140,45)
%\node[Nframe=n](label)(45,5){Th.~\ref{theo1}}

\drawline[AHnb=0, arcradius=5](2,50)(2,10)(85,10)(100,28)(145,28)(145,50)
\node[Nframe=n](label)(50,5){Th.~\ref{theo1}}

\drawline[AHnb=0, arcradius=5](4,50)(4,23)(85,23)(100,15)(147,15)(147,50)
\node[Nframe=n](label)(125,10){Th.~\ref{thrm:partialobs-reduction}}

%\drawcurve[AHnb=0](5,50)(6,25)(12,24)(82,22)(87,21)(97,20)(102,17)(135,15)(140,16)(142,45)
%\node[Nframe=n](label)(120,10){Th.~\ref{thrm:partialobs-reduction}}

%\drawloop[ELside=l,loopCW=y, loopdiam=6](q0){$b$}
%\drawloop[ELside=l,loopCW=y, loopdiam=6](q1){$b$}
%\drawloop[ELside=l,loopCW=y, loopdiam=6](q2){$b$}
%\drawloop[ELside=l,loopCW=y, loopdiam=6](q3){$b$}
%\drawloop[ELside=l,loopCW=y, loopdiam=6](q4){$b$}
%\drawloop[ELside=l,loopCW=y, loopangle=0, loopdiam=6](q6){$a,b$}

%\drawedge[dash={1}0](n3bis,nkbis){$0,1$}
\end{picture}

  \caption{Hierarchy of the various classes of game graphs. 
According to Theorem~\ref{theo1} randomness is for free in the
transition function for concurrent games even with complete observation,
and according to Theorem~\ref{thrm:partialobs-reduction} randomness is for free 
in the transition function for one-sided complete observation games even 
if they are turn-based.
%The curves materialize the classes for which 
%randomness is for free in transition relation (Theorem~\ref{theo1} and Theorem~\ref{thrm:partialobs-reduction}).
For $2\half$-player games, randomness in the
transition function is not for free only in complete-observation turn-based games.
  \label{figure-overview}}
   \end{center}
\end{figure}

%\vspace{-1.5em}
\section{Randomness for Free in Transition Function}\label{sec:r4free-transitions}
In this section we present a precise characterization of the classes of 
games where randomness in the transition function can be obtained for \emph{free}:
in other words, we present the precise characterization of classes of 
games with probabilistic transition function that can be reduced to 
the corresponding class with deterministic transition function. 
We present our results as three reductions: (a) the first reduction 
allows us to separate probability from the mode of interaction; 
(b) the second reduction shows how to simulate probability in transition 
function with \PC\ (complete-observation concurrent) deterministic transition function;
and (c) the final reduction shows how to simulate probability in 
transition with \OT\ (one-sided complete-observation turn-based) deterministic 
transition function. We then show that for \PT\ (complete-observation turn-based) games,
randomness in the transition function cannot be obtained for free, and conclude with 
the \emph{concurrency for free} result that \OT\ and \IT\ games 
can simulate \OC\ and \IC\ games respectively.

A \emph{reduction} from a class $\calg$ of games to a class $\calg'$ 
is a mapping that, from a game $G \in \calg$ and an objective $\varphi$ in $G$,
returns a game $G' \in \calg'$ and an objective $\varphi'$ in $G'$, and such that 
the state space $S$ of $G$ is (injectively) mapped to the state space $S'$ of $G'$. 
In all our reductions we have $S \subseteq S'$, and thus 
the state-space mapping is the identity (on $S$). The mapping of objectives
in our reductions is such that $\varphi$ is the projection of $\varphi'$ on $S^{\omega}$. 
It follows that when $\varphi$ is a parity objective defined with at most $d$ priorities, 
then so is $\varphi'$ (and in the sequel, we omit the definition of the priority
function for $\varphi'$), and when $\varphi$ is an objective in the $k$-th level of the Borel 
hierarchy, then so is $\varphi'$.

All our reductions are \emph{local}: they consist of a 
gadget construction and replacement locally at every state.
Additional properties of interest for reductions are as follows:
\begin{itemize}
\item 
A reduction is \emph{almost-sure-preserving} (resp., \emph{positive-preserving}),  
if for all states $s \in S$ in $G$: player~$1$ is almost-sure winning (resp., positive winning) 
in $G$ from~$s$ if and only if player~$1$ is almost-sure winning (resp., positive winning) 
in $G'$ from~$s$.

\item
A reduction is \emph{value-preserving} if $\va^G(\varphi)(s) = \va^{G'}(\varphi')(s)$ for all $s \in S$, 
and \emph{threshold-preserving} if for all $\eta \in \real$, all states $s \in S$, 
and all $\bowtie\,\in\! \{>,\geq\}$:
there exists an observation-based strategy $\straa \in \Straa_G^O$ for player~$1$ in $G$ such that 
$\forall \strab \in \Strab_G^O:  \Prb_{s}^{\straa,\strab}(\varphi) \bowtie \eta$ if and only if 
there exists an observation-based strategy $\straa' \in \Straa_{G'}^O$ for player~$1$ in $G'$ such that 
$\forall \strab' \in \Strab_{G'}^O:  \Prb_{s}^{\straa',\strab'}(\varphi') \bowtie \eta$.

\end{itemize}

% \begin{figure}[p]
%    \begin{center}
%       \input{figures/overview2.tex}
%    \end{center}
%   \caption{   \label{figure-overview2}}
% \end{figure}

\begin{figure}[!tb]
   \begin{center}
      \unitlength=.6mm

\begin{picture}(102,65)(0,0)
\put(0,0){\framebox(102,65){}}

\gasset{Nw=9,Nh=9,Nmr=4.5,rdist=1,Nadjust=n}

\node[Nframe=n](q4)(25,55){\IC\/ $\equiv$ \IT\/}
\node[Nframe=n](q4)(16,55){}
\node[Nframe=n](q3)(25,40){\OC\/ $\equiv$ \OT\/}
\node[Nframe=n](q3)(16,40){}
\node[Nframe=n](q2)(16,25){\PC\/}
\node[Nframe=n](q1)(16,10){\PT\/}

\drawedge[AHnb=0](q1,q2){}
\drawedge[AHnb=0](q2,q3){}
\drawedge[AHnb=0](q3,q4){}

\node[Nframe=n](label)(48,47.5){$\left.\begin{array}{l}\phantom{1}\\[+12pt] \phantom{2}\end{array}\right\}$ Th.~\ref{thrm:concurrency4free}}
%\node[Nframe=n](label)(60,52.5){Th.~\ref{thrm:concurrency4free}}
\node[Nframe=n](label)(74,39.5){Concurrency for free}

\drawline[AHnb=0, arcradius=4, dash={1}{0}](6,60)(6,18)(25,18)(42,33)(42,60)
\node[Nframe=n](label)(51,21){Th.~\ref{theo1} \& Th.~\ref{thrm:partialobs-reduction}}
\node[Nframe=n](label)(59,13){Randomness for free}

%\drawline[AHnb=0](3,40)(5,20)(85,20)(95,10)(145,10)(148,30)

%\drawcurve[AHnb=0](3,50)(5,12)(10,11)(80,11)(87,12)(100,24)(102,25)(132,26)(139,27)(140,45)
%\node[Nframe=n](label)(45,5){Th.~\ref{theo1}}

%\drawcurve[AHnb=0](5,50)(6,25)(12,24)(82,22)(87,21)(97,20)(102,17)(135,15)(140,16)(142,45)
%\node[Nframe=n](label)(120,10){Th.~\ref{thrm:partialobs-reduction}}

%\drawloop[ELside=l,loopCW=y, loopdiam=6](q0){$b$}
%\drawloop[ELside=l,loopCW=y, loopdiam=6](q1){$b$}
%\drawloop[ELside=l,loopCW=y, loopdiam=6](q2){$b$}
%\drawloop[ELside=l,loopCW=y, loopdiam=6](q3){$b$}
%\drawloop[ELside=l,loopCW=y, loopdiam=6](q4){$b$}
%\drawloop[ELside=l,loopCW=y, loopangle=0, loopdiam=6](q6){$a,b$}

%\drawedge[dash={1}0](n3bis,nkbis){$0,1$}
\end{picture}

  \caption{Summary of the results of Section~\ref{sec:r4free-transitions}.   \label{figure-overview3}}
   \end{center}
\end{figure}
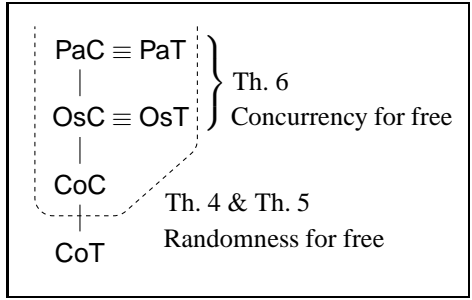

All reductions presented in this paper are threshold-preserving. Note that 
threshold-preserving implies value-preserving, almost-sure-preserving ($\bowtie\: =\: \geq$, $\eta = 1$), 
and positive-preserving ($\bowtie\: =\: >$, $\eta = 0$).

A reduction \emph{restriction-preserving} 
if when $G$ is one-sided complete-observation, then so is $G'$, 
when $G$ is complete-observation, then so is $G'$, and when $G$ is turn-based, then so is $G'$.
We say that a reduction is computable in \emph{polynomial time} (resp., 
in \emph{exponential time}) if the game $G'$ can be constructed 
in polynomial time (resp., in exponential time) from $G$ (assuming a reasonable 
encoding of games, such as explicit lists of binary-encoded states, observations, 
actions, and transitions, and rational probabilities encoded in binary).

%Our reductions preserve values, existence of $\vare$-optimal strategies 
%for $\vare\geq0$, and also existence of almost-sure and positive winning 
%strategies. 

An overview of the class of games for which randomness is for free
in the transition function (which we establish in this section) is 
given in \figurename~\ref{figure-overview3}.

\subsection{Separation of probability and interaction}
A concurrent game of partial observation $G$ satisfies the 
\emph{interaction separation} condition if the following restrictions are satisfied
(see also \figurename~\ref{figure-example-inter}):
the state space $S$ can be partitioned into $(S_A,S_P)$ such that 
(1) $\trans: S_A \times A_1 \times A_2 \to S_P$, and 
(2) $\trans: S_P \times A_1 \times A_2 \to \dist(S_A)$ such 
that for all $s \in S_P$ and all $s' \in S_A$, and for all 
$a_1,a_2,a_1',a_2'$ we have 
$\trans(s,a_1,a_2)(s')=\trans(s,a_1',a_2')(s')=\trans(s,-,-)(s')$.
In other words, the choice of actions (or the interaction) of the players 
takes place at states in $S_A$ and actions determine a unique 
successor state in $S_P$, and the transition function at $S_P$ is 
probabilistic and independent of the choice of the players.
In this section, we present a reduction of each class of games to the corresponding 
class satisfying interaction separation, and we present a reduction
to games with uniform transition probabilities.

\smallskip\noindent{\bf Reduction to interaction separation.} 
Let $G=\tuple{S,A_1,A_2,\trans,\Obs_1,\Obs_2}$ be a 
concurrent game of partial observation with an objective $\varphi$.
We obtain a concurrent game of partial observation 
$\ov{G} = \tuple{S_A \cup S_P, A_1, A_2, \ov{\trans}, \ov{\Obs}_1,\ov{\Obs}_2}$ 
where $S_A = S$, $S_P  = S\times A_1 \times A_2$, and:
\begin{itemize}
\item \emph{Observations.} For $i \in \set{1,2}$, if $\Obs_i=\{ \{s\} \mid s \in S \}$, 
then $\ov{\Obs}_i = \{ \{s'\} \mid s' \in S_A \cup S_P \}$; 
otherwise $\ov{\Obs}_i = \{ o \cup o \times A_1 \times A_2 \mid o \in \Obs_i \}$.
\item \emph{Transition function.} The transition function is as follows:
\begin{enumerate}
\item We have the following three cases: 
(a) if $s$ is a player~1 turn-based state, then pick an action $a_2^*$ and 
for all $a_2$ let $\ov{\trans}(s,a_1,a_2)=(s,a_1,a_2^*)$;
(b) if $s$ is a player~2 turn-based state, then pick an action $a_1^*$ and 
for all $a_1$ let $\ov{\trans}(s,a_1,a_2)=(s,a_1^*,a_2)$; and
(c) otherwise, $\ov{\trans}(s,a_1,a_2)=(s,a_1,a_2)$;
\item for all $(s,a_1,a_2) \in S_P$ we have
$\ov{\trans}((s,a_1,a_2),-,-)(s')=\trans(s,a_1,a_2)(s')$.
\end{enumerate}
%Thus the states in $S$ are $S_A$ where the interaction takes places,
%and the states in $S \times A_1 \times A_2$ are the purely probabilistic states $S_P$.

\item \emph{Objective mapping.} Given the objective $\varphi$ in $G$ we 
obtain the objective $\ov{\varphi}=\set{s_0 s_0' s_1 s_1' \ldots \mid 
s_0 s_1 \ldots \in \varphi }$ in $\ov{G}$.
\end{itemize}
It is easy to map observation-based strategies of the game $G$ to observation-based 
strategies in $\ov{G}$ and vice-versa to preserve satisfaction of $\varphi$ and 
$\ov{\varphi}$ in $G$ and $\ov{G}$, respectively.
%Let us refer to the above reduction as $\reduction$: i.e., $\reduction(G, \varphi) =(\ov{G},\ov{\varphi})$.
Then we have the following theorem.

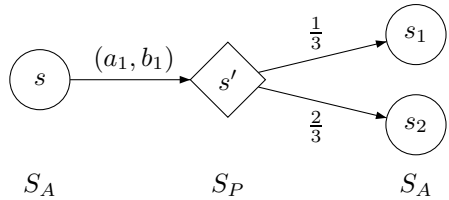
\begin{figure}[t]
   \begin{center}
   \hrule
      \begin{picture}(60,30)(0,0)
%\put(0,-2){\framebox(60,30){}}

%\gasset{Nw=9,Nh=9,Nmr=4.5,rdist=1, loopdiam=6}

\node[Nmarks=n](n0)(5,17){$s$}
\rpnode[Nmarks=n](n1)(30,17)(4,5){$s'$}
%\rpnode[Nmarks=n](n2)(30,14)(4,5){$s'_1$}

\node[Nmarks=n](n3)(55,23){$s_1$}
\node[Nmarks=n](n4)(55,11){$s_2$}

\node[Nframe=n](lab)(5,3){$S_A$}
\node[Nframe=n](lab)(30,3){$S_P$}
\node[Nframe=n](lab)(55,3){$S_A$}

%\drawloop[ELside=l,loopCW=y](nk){$0,1$}

\drawedge[ELpos=50, ELside=l, ELdist=1, curvedepth=0](n0,n1){$(a_1,b_1)$}
%\drawedge[ELpos=50, ELside=r, ELdist=-1, curvedepth=0](n0,n2){\begin{tabular}{ll}$(a_1,b_2)$\\$(a_2,b_1)$\end{tabular}}

\drawedge[ELpos=50, ELside=l, curvedepth=0](n1,n3){$\frac{1}{3}$}
\drawedge[ELpos=50, ELside=r, curvedepth=0](n1,n4){$\frac{2}{3}$}

\end{picture}
   \hrule
  \vspace{6pt}
  \caption{Example of interaction separation for $\trans(s,a_1,b_1)(s_1) = \frac{1}{3}$ 
and $\trans(s,a_1,b_1)(s_2) = \frac{2}{3}$.
  \label{figure-example-inter}}
   \end{center}
\end{figure}

\begin{comment}
\begin{theorem}\label{theo:separation}
Let $G$ be a concurrent game of partial observation with an objective 
$\varphi$, and let $(\ov{G},\ov{\varphi})=\reduction(G,\varphi)$.
Then the following assertions hold:
\begin{compactenum}
\item The reduction $\reduction$ is \emph{restriction preserving}:
if $G$ is one-sided complete-observation, then so is $\ov{G}$; 
if $G$ is complete-observation, then so is $\ov{G}$; 
if $G$ is turn-based, then so is $\ov{G}$.

\item For all $s \in S$, there is an observation-based almost-sure (resp. positive) winning strategy 
for $\varphi$ from $s$ in $G$ iff there is an observation-based almost-sure (resp. 
positive) winning strategy for $\ov{\varphi}$ from $s$ in $\ov{G}$.

\item The reduction is objective-preserving: 
if $\varphi$ is a parity objective, then so is $\ov{\varphi}$; 
if $\varphi$ is an objective in the $k$-the level of the Borel hierarchy, then
so is $\ov{\varphi}$.

\item For all $s \in S$ we have
%\[
$\va(\varphi)(s)= \va^{\ov{G}}(\ov{\varphi})(s)$. 
%\]
For all $s \in S$
there is an observation-based optimal strategy for $\varphi$ from $s$ 
in $G$
iff there is an observation-based optimal strategy for $\ov{\varphi}$ 
from $s$ in $\ov{G}$.
\end{compactenum}
\end{theorem}
\end{comment}

\begin{theorem}\label{theo:separation}
There exists a reduction from the class of partial-observation concurrent games (\IC\/ games)
to the class of \IC\/ games with interaction separation such that this reduction is
\begin{compactenum}
\item threshold-preserving,
\item restriction-preserving, and
\item computable in polynomial time.
%\item restriction-preserving and objective-preserving, 
%\item computable in polynomial time,
%\item value-preserving, and threshold-preserving (and thus also positive- and almost-sure-preserving).
\end{compactenum}
\end{theorem}

Since the reduction is restriction-preserving,  
we have a reduction that separates the interaction and probabilistic
transition maintaining the restriction of observation and mode of interaction.

\smallskip\noindent{\bf Uniform-$n$-ary concurrent games.} 
The class of \emph{uniform-$n$-ary games} is the 
special class of games satisfying interaction separation and such 
that for every state $s \in S_P$ the probability $\trans(s,-,-)(s')$ to 
a successor state $s'$ is a multiple of $\frac{1}{n}$.
It follows from the results of~\cite{ZP95} that every \PC\ game 
with rational transition probabilities can be reduced in polynomial time 
to an equivalent polynomial-size uniform-binary (i.e., $n=2$) 
\PC\ game for all parity objectives. The reduction is 
achieved by adding dummy states to simulate the probability, and the 
reduction extends to all objectives (in the reduced game 
we need to consider the objective whose projection in the original 
game gives the original objective).

In the case of partial information, the reduction to uniform-binary games
of~\cite{ZP95} does not work.
To see this, consider \figurename~\ref{fig:uniform-binary} %%%(in appendix) 
where two probabilistic 
states $s_1,s_2$ have the same observation (i.e., $\obs_1(s_1) = \obs_1(s_2)$)
and the outgoing probabilities are $\tuple{\frac{1}{4},\frac{3}{4}}$ from $s_1$
and $\tuple{\frac{1}{3},\frac{2}{3}}$ from $s_2$. The corresponding uniform-binary game 
(given in \figurename~\ref{fig:uniform-binary})
is not equivalent to the original game because the number of steps needed to 
simulate the probabilities is not always the same from $s_1$ and from $s_2$. 
From $s_1$ two steps are always sufficient, while from $s_2$ more than two steps 
may be necessary (with probability $\frac{1}{4}$). 
Therefore with probability $\frac{1}{4}$, player~1 observing more than 2 steps would
infer that the game was for sure in $s_2$, thus artificially improving his knowledge 
and increasing his value function.

\begin{figure}[!tb]
  \begin{center}
    \hrule
    \begin{picture}(100,69)(0,0)
%\put(0,0){\framebox(100,69){}}

%\gasset{Nw=9,Nh=9,Nmr=4.5,rdist=1, loopdiam=6}

\rpnode[Nmarks=n](n0)(10,52)(4,5){$s_1$}
\node[Nmarks=n](n1)(30,62){$s'_1$}
\node[Nmarks=n](n2)(30,42){$s''_1$}

\drawedge[ELpos=50, ELside=l, curvedepth=0](n0,n1){$\frac{1}{4}$}
\drawedge[ELpos=50, ELside=r, curvedepth=0](n0,n2){$\frac{3}{4}$}

%\drawloop[ELside=l,loopCW=y](nk){$0,1$}

\rpnode[Nmarks=n](n0)(10,17)(4,5){$s_2$}
\node[Nmarks=n](n1)(30,27){$s'_2$}
\node[Nmarks=n](n2)(30,7){$s''_2$}

\drawedge[ELpos=50, ELside=l, curvedepth=0](n0,n1){$\frac{1}{3}$}
\drawedge[ELpos=50, ELside=r, curvedepth=0](n0,n2){$\frac{2}{3}$}

\node[Nmarks=n, Nw=12, Nh=47, Nmr=3, dash={1.5}0, ExtNL=y, NLangle=270, NLdist=2](A1)(10,34.5){}

\rpnode[Nmarks=n](n0)(50,52)(4,5){$s_1$}
\rpnode[Nmarks=n](nA)(70,62)(4,5){}
\rpnode[Nmarks=n](nB)(70,42)(4,5){}
\node[Nmarks=n](n1)(90,62){$s'_1$}
\node[Nmarks=n](n2)(90,42){$s''_1$}

\drawedge[ELpos=50, ELside=l, ELdist=0.5, curvedepth=0](n0,nA){$\frac{1}{2}$}
\drawedge[ELpos=50, ELside=r, curvedepth=0](n0,nB){$\frac{1}{2}$}

\drawedge[ELpos=50, ELside=l, curvedepth=0](nA,n1){$\frac{1}{2}$}
\drawedge[ELpos=50, ELside=r, curvedepth=0](nA,n2){$\frac{1}{2}$}
\drawedge[ELpos=50, ELside=r, curvedepth=0](nB,n2){$1$}

%\drawloop[ELside=l,loopCW=y](nk){$0,1$}

\rpnode[Nmarks=n](n0)(50,17)(4,5){$s_2$}
\rpnode[Nmarks=n](nA)(70,27)(4,5){}
\rpnode[Nmarks=n](nB)(70,7)(4,5){}
\node[Nmarks=n](n1)(90,27){$s'_2$}
\node[Nmarks=n](n2)(90,7){$s''_2$}

\drawedge[ELpos=50, ELside=l, curvedepth=0](n0,nA){$\frac{1}{2}$}
\drawedge[ELpos=50, ELside=r, curvedepth=0](n0,nB){$\frac{1}{2}$}

\drawedge[ELpos=50, ELside=l, curvedepth=0](nA,n1){$\frac{1}{2}$}
\drawedge[ELpos=50, ELside=r, curvedepth=0](nA,n2){$\frac{1}{2}$}
\drawedge[ELpos=50, ELside=r, curvedepth=0](nB,n2){$\frac{1}{2}$}
\drawedge[ELpos=72, ELside=l, curvedepth=8](nB,n0){$\frac{1}{2}$}

%\drawline[AHnb=0, arcradius=2, dash={1.5}0](45,45)(45,55)(55,55)(55,30)(75,30)(75,0)(45,0)(45,45)
\node[Nmarks=n, Nw=32, Nh=67, Nmr=3, dash={1.5}0, ExtNL=y, NLangle=270, NLdist=2](A1)(60,34.5){}

%\drawedge[dash={1}0](n3bis,nkbis){$0,1$}

\end{picture}
    \hrule
  \vspace{6pt}
    \caption{An example showing why the uniform-binary reduction cannot be used with partial observation. \label{fig:uniform-binary}}
  \end{center}
\end{figure}
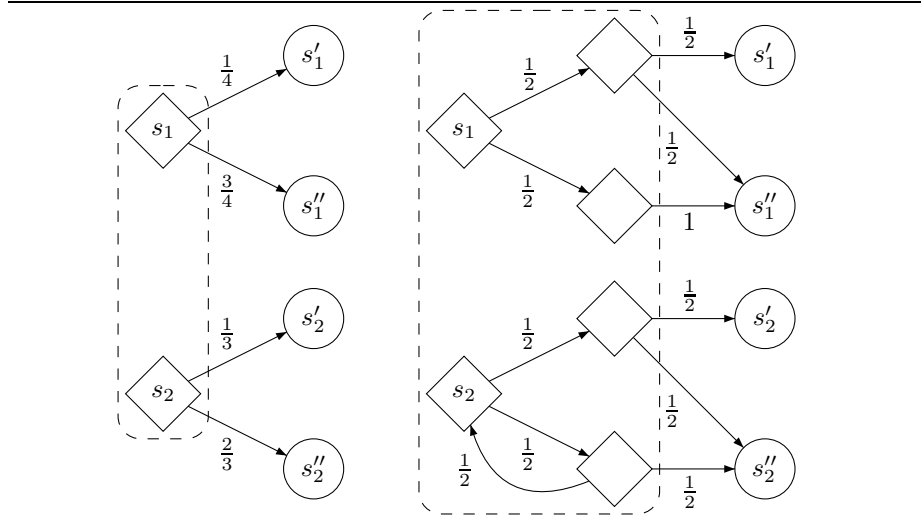

Therefore in the case of a partial-observation game~$G$ satisfying interaction separation, 
we present a reduction to a uniform-$n$-ary game $G'$ where $n = 1/r$ where $r$ is the greatest 
common divisor of all probabilities in the original game~$G$ (a rational~$r$ is a divisor 
of a rational~$p$ if $p=q \cdot r$ for some integer $q$).
Note that the number $n=1/r$ is an integer. We denote by $[n]$ the set $\{0,1,\dots,n-1\}$.
For a probabilistic state $s \in S_P$, we define the $n$-tuple 
$\Succ(s) = \tuple{s'_0,\dots,s'_{n-1}}$ in which each state $s' \in S$
occurs $n \cdot \trans(s,-,-)(s')$ times. Then, we can view the transition 
relation $\trans(s,-,-)$ as a function assigning the same probability $r = 1/n$ 
to each element of $\Succ(s)$ (and then adding up the probabilities of identical elements). 
Hence it is straightforward to obtain a uniform-$n$-ary game $G'$.

\begin{theorem}\label{theo:uniform}
There exists a reduction from the class of \IC\/ games
to the class of uniform-$n$-ary \IC\/ games (where $1/n$ is the greatest 
common divisor of all probabilities in the original game) such that this reduction is
\begin{compactenum}
\item threshold-preserving,
\item restriction-preserving, and
\item computable in exponential time (and in polynomial time for \PC\ games~\cite{ZP95}).
%\item restriction-preserving and objective-preserving, 
%\item computable in polynomial time,
%\item value-preserving, and threshold-preserving (and thus also positive- and almost-sure-preserving).
\end{compactenum}
\end{theorem}

Note that the above reduction is worst-case exponential (because so can be the 
inverse of the greatest common divisor of the transition probabilities). 
This is necessary to have the property that 
all probabilistic states in the game have the same number of successors. %%We will see that 
This property is crucial because it determines the number of actions available to player~1 
in the reductions presented in Section~\ref{sec:red1} and~\ref{sec:red2}, 
and the number of available actions should not differ in states that have 
the same observation.

\subsection{Simulating probability by complete-observation concurrent determinism}\label{sec:red1}
In this section, we show that probabilistic states can be simulated by \PC\ 
deterministic gadgets (and hence also by \OC\  and \IC\  deterministic 
gadgets).
By Theorem~\ref{theo:separation} and Theorem~\ref{theo:uniform}, we focus on 
uniform-$n$-ary games.
A probabilistic state with uniform probability over the successors 
is simulated by a complete-observation concurrent 
deterministic state where the optimal strategy for both players is to 
play uniformly over the set of available actions.
% (more details are given in~\cite{OurTechRpt}). %This gives us Theorem~\ref{theo1}.

\begin{theorem}\label{theo1}
Let $a \in \set{\I,\Os,\P}$ and $b\in \set{\C,\T}$, and let 
$\calc=ab$ and $\calc'=a\C$. 
There exists a reduction from the class of games $\calg_\calc$ to the
class of games $\calg_{\calc'} \cap \calg_D$ (thus with deterministic
transition function) such that this reduction is
\begin{compactenum}
\item threshold-preserving, and
\item computable in polynomial time if $a=\P$, and~in exponential time if $a=\I$ or $a=\Os$.
%\item objective-preserving, 
%\item computable in polynomial time if $a=\P$, and~in exponential time if $a=\I$ or $a=\Os$,
%\item value-preserving and threshold-preserving (and thus also positive- and almost-sure-preserving).
\end{compactenum}
\end{theorem}

\begin{comment}
\begin{theorem}\label{theo1}
Let $a \in \set{\I,\Os,\P}$ and $b\in \set{\C,\T}$, and let 
$\calc=ab$ and $\calc'=a\C$. 
Let $G$ be a game in $\calg_\calc$ 
with probabilistic transition function with rational probabilities and an objective $\varphi$. 
A game $\ov{G} \in \calg_{\calc'}\cap \calg_D$ (in the class that subsumes 
$\calg_\calc$ with concurrent interaction) with deterministic transition 
function can be constructed in (a)~polynomial time if $a=\P$, and
(b)~in exponential time if $a=\I$ or $\Os$, 
with an objective $\ov{\varphi}$ such that the state space of $G$ is a subset of 
the state space of $\ov{G}$ and we have: 
\begin{compactenum}
\item For all $s \in S$ there is an observation-based almost-sure (resp. positive) winning strategy 
from $s$ for $\varphi$ in $G$ iff there is an observation-based almost-sure (resp. 
positive) winning strategy for $\ov{\varphi}$ from $s$ in $\ov{G}$.

\item For all $s \in S$ we have
%\[
$\va(\varphi)(s)= \va^{\ov{G}}(\ov{\varphi})(s)$. 
%%\]
For all $s \in S$ there is an observation-based optimal strategy 
for $\varphi$ from $s$ in $G$
iff there is an observation-based optimal strategy for $\ov{\varphi}$ 
from $s$ in $\ov{G}$.
\end{compactenum}
\end{theorem}
\end{comment}

\begin{proof} To prove the result we show that a uniform-$n$-ary 
probabilistic state can be simulated by a \PC\ deterministic gadget.
For simplicity we present the details for the case when $n=2$, 
and the gadget for the general case is presented later.
%given in the Appendix.
Our reduction is as follows: we consider a uniform-binary \PC\ game 
such that there is only one probabilistic state, and reduce it to a 
\PC\ deterministic game. For uniform-binary \PC\ games 
with multiple probabilistic states the reduction can be applied to each
state one at a time and we would obtain the desired reduction from 
uniform-binary \PC\ games to \PC\ deterministic games. 
It is easy to see that the reduction can be computed in polynomial time
from uniform-$n$-ary games. The complexity result (item (2) of the theorem)
then follows from Theorem~\ref{theo:separation} and Theorem~\ref{theo:uniform}.

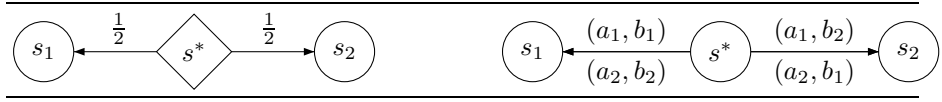
\begin{figure}[!tb]
  \begin{center}
    %\unitlength=4pt
    \hrule
    \begin{picture}(125,12)(0,0)
    %\put(0,0){\framebox(125,12){}}    
    %\gasset{Nw=5,Nh=5,Nmr=2.5,curvedepth=0}
    %\thinlines
    \rpnode(A1)(25,6)(4,5){$s^*$}
    \node(A2)(5,6){$s_1$}
    \node(A3)(45,6){$s_2$}
    \drawedge[ELside=r](A1,A2){$\frac{1}{2}$}
    \drawedge[ELside=l](A1,A3){$\frac{1}{2}$}
    
    \node(A4)(95,6){$s^*$}
    \node(A5)(70,6){$s_1$}
    \node(A6)(120,6){$s_2$}
    \drawedge[ELpos=50, ELside=r](A4,A5){$(a_1,b_1)$}
    \drawedge[ELpos=50, ELside=l](A4,A5){$(a_2,b_2)$}
    \drawedge[ELpos=50, ELside=l](A4,A6){$(a_1,b_2)$}
    \drawedge[ELpos=50, ELside=r](A4,A6){$(a_2,b_1)$}
\end{picture}
    \hrule
  \vspace{6pt}
  \caption{The reduction of uniform-binary \PC\ games. \label{fig:red}}
  \end{center}
\end{figure}

%Hence we prove the following claim.

%\smallskip\noindent{\em Claim.}
The reduction is illustrated in \figurename~\ref{fig:red} and is defined as follows.
Consider a uniform-binary \PC\ game $G$ with a single
probabilistic state $s^*$ with two successors $s_1$ and $s_2$.
Construct the \PC\ deterministic game $G'$ obtained from~$G$ by 
transforming the state $s^*$ to a concurrent deterministic state as follows:
the actions available for player~1 at $s^*$ are $a_1$ and $a_2$, and the 
actions available for player~2 at $s^*$ are $b_1$ and $b_2$; the transition 
function is as follows:
$\trans(s^*,a_1,b_1)=\trans(s^*,a_2,b_2)=s_1$ and 
$\trans(s^*,a_1,b_2)=\trans(s^*,a_2,b_1)=s_2$.
Note that the state space of $G'$ is the same as in $G$, thus $\varphi' = \varphi$.
Then for all objectives $\varphi$, we show that the reduction is 
threshold-preserving as follows.

\begin{enumerate}
\item First assume that there exists an observation-based strategy $\straa$ for player~1 in $G$ such that 
$\forall \strab \in \Strab_G^O:  \Prb_{s}^{\straa,\strab}(\varphi) \bowtie \eta$
for some arbitrary $\eta \in \real$, $s \in S$, and $\bowtie\,\in\! \{>,\geq\}$, and construct
a strategy $\straa'$ for player~1 in $G'$ as follows: 
the strategy $\straa'$ copies
the strategy $\straa$ for all histories other than when the current state is $s^*$,
and if the current state is $s^*$, then the strategy $\straa'$ plays the actions 
$a_1$ and $a_2$ uniformly with probability $\frac{1}{2}$. 
Given the strategy $\straa'$, if the current state is $s^*$, then for any probability
distribution over player~$2$'s actions $b_1$ and $b_2$, the successor states are $s_1$ and $s_2$ 
with probability~$\frac{1}{2}$ (i.e., it plays exactly the role of state $s^*$ in $G$).
It follows that for all strategies $\strab'$ of player~$2$ in $G'$, there is a strategy
$\strab$ in $G$ (that plays like $\strab'$ for all histories in $G$) such that 
$\Prb_{s}^{\straa,\strab}(\varphi) = \Prb_{s}^{\straa',\strab'}(\varphi)$
and thus $\Prb_{s}^{\straa',\strab'}(\varphi) \bowtie \eta$.

\item Second assume that there exists an observation-based strategy $\straa'$ for 
player~1 in $G'$ such that $\forall \strab' \in \Strab_{G'}^O: \Prb_{s}^{\straa',\strab'}(\varphi) \bowtie \eta$
for some arbitrary $\eta \in \real$, $s \in S$, and $\bowtie\,\in\! \{>,\geq\}$, 
and consider the strategy $\straa$ for player~1 in $G$ that plays like $\straa'$ 
for all histories in $G$. Assume towards contradiction that against $\straa$ 
there exists a strategy $\strab \in \Strab_{G}^O$ such that 
$\lnot \Prb_{s}^{\straa,\strab}(\varphi) \bowtie \eta$. Then consider the 
strategy $\strab'$ in $G'$ that  copies the strategy $\strab$ for all histories 
other than when the current state is $s^*$, and if the current state is $s^*$, 
then the strategy $\strab'$ plays the actions $b_1$ and $b_2$ uniformly with 
probability $\frac{1}{2}$. Given the strategy $\strab'$ in $G'$, if the current 
state is $s^*$, then for any probability distribution over player~$1$'s actions 
$a_1$ and $a_2$, the successor states are $s_1$ and $s_2$ with probability~$\frac{1}{2}$ 
(i.e., it plays exactly the role of state $s^*$ in $G$). It follows that 
$\Prb_{s}^{\straa',\strab'}(\varphi) = \Prb_{s}^{\straa,\strab}(\varphi)$
and thus $\lnot \Prb_{s}^{\straa',\strab'}(\varphi) \bowtie \eta$, in contradiction
with the assumption on $\straa'$. Therefore, such a strategy $\strab$ cannot exist,
and we have $\Prb_{s}^{\straa,\strab}(\varphi) \bowtie \eta$ for all $\strab \in \Strab_G^O$,
which concludes the proof that the reduction is threshold-preserving.
\end{enumerate}

\medskip\noindent{\itshape Gadget for uniform-$n$-ary probability reduction.}
%for Theorem~\ref{theo1}.} 
We now show how to simulate a probabilistic state $s^*$, with $n$ 
successors $s_0,s_1,\ldots,s_{n-1}$ such that the transition 
probability is $1/n$ to each of the successors, by a concurrent 
deterministic state.
In the concurrent deterministic state $s^*$ there are $n$ actions 
$a_0,a_1,\ldots,a_{n-1}$ available for player~1 and $n$ actions 
$b_0,b_1,\ldots,b_{n-1}$ available for player~2.
The transition function is as follows: for $0 \leq i < n$ and 
$0 \leq j < n$ we have $\trans(s^*,a_i,b_j)=s_{(i+j)\mod n}$.
Intuitively, the transition function matrix is obtained as follows:
the first row is filled with states $s_0,s_1, \ldots, s_{n-1}$, and 
from a row $i$, the row $i+1$ is obtained by moving the state of the
first column of row $i$ to the last column in row $i+1$ and left-shifting
by one position all the other states; the construction is illustrated 
on an example with $n=4$ successors in~(\ref{trans:matrix}). 
The construction ensures that in every row and every column 
each state $s_0,s_1,\ldots,s_{n-1}$ appears exactly once. 
It follows that if player~1 plays all actions uniformly at random, then against 
any probability distribution of player~2 the successor states are 
$s_0,s_1, \ldots, s_{n-1}$ with probability $1/n$ each; 
and a similar result holds if player~2 plays all actions uniformly at random.
The correctness of the reduction for uniform-$n$-ary probabilistic state
is then exactly as for the case of $n=2$. %%the proof of Theorem~\ref{theo1}.

\begin{equation}\label{trans:matrix}
\begin{bmatrix}
s_0 & s_1 & s_2 & s_3 \\
s_1 & s_2 & s_3 & s_0 \\
s_2 & s_3 & s_0 & s_1 \\
s_3 & s_0 & s_1 & s_2 
\end{bmatrix}
\end{equation}

The desired result follows.
\qed
\end{proof}

\subsection{Simulating probability by one-sided complete-observation turn-based determinism}\label{sec:red2}

We show that probabilistic states can be simulated by \OT\ (one-sided complete-observation turn-based) 
states, and by Theorem~\ref{theo:separation} we consider games that satisfy interaction separation. 
The reduction is illustrated in \figurename~\ref{fig:gamify}: 
each probabilistic state $s$ is transformed into a player-$2$ state with $n$ successor 
player-$1$ states (where $n$ is chosen such that the probabilities from~$s$ are integer 
multiples of $1/n$, in the example $n=3$).
Because all successors of $s$ have the same observation, player~$1$ has no advantage in playing 
after player~$2$, and because by playing all actions uniformly at random each player
can unilaterally decide to simulate the probabilistic state, the value and 
properties of strategies of the game are preserved. 
%%Due to lack of space, the proof of Theorem~\ref{thrm:partialobs-reduction} is given in the appendix.

\begin{comment}
\begin{theorem}\label{thrm:partialobs-reduction}
Let $a \in \set{\I,\Os,\P}$ and $b\in \set{\C,\T}$, and 
let $a'=a$ if $a \neq \P$, and $a'=\Os$ otherwise. 
Let $\calc=ab$ and $\calc'=a'b$. 
Let $G$ be a game in $\calg_\calc$ 
with probabilistic transition function with rational transition probabilities and an objective $\varphi$. 
A game $G' \in \calg_{\calc'} \cap \calg_D$ (in the class that subsumes 
one-sided complete-observation turn-based games and the class $\calg_\calc$) with deterministic transition 
function can be constructed in exponential time with an objective $\varphi'$ 
such that the state space of $G$ is a subset of the state space of $G'$ and we have:
\begin{compactenum}
\item For all $s \in S$ there is an observation-based almost-sure (resp. positive) winning strategy 
from $s$ for $\varphi$ in $G$ iff there is an observation-based almost-sure (resp. 
positive) winning strategy for $\varphi'$ from $s$ in $G'$.

\item For all $s \in S$ we have
%%\[
$\va(\varphi)(s)= \va^{G'}(\varphi')(s)$. 
%%\]
For all $s \in S$ there is an observation-based optimal strategy 
for $\varphi$ from $s$ in $G$
iff there is an observation-based optimal strategy for $\varphi'$ 
from $s$ in $G'$.
\end{compactenum}
\end{theorem}
\end{comment}

%\mynote{L: I modified the next theorem.}

\begin{theorem}\label{thrm:partialobs-reduction}
Let $a \in \set{\I,\Os,\P}$ and $b\in \set{\C,\T}$, and 
let $a'=\Os$ if $a = \P$, and $a'=a$ otherwise. 
Let $\calc=ab$ and $\calc'=a'b$. 
There exists a reduction from the class of games $\calg_\calc$ to the
class of games $\calg_{\calc'} \cap \calg_D$ (thus with deterministic
transition function) such that this reduction is
\begin{compactenum}
\item threshold-preserving, and
\item computable in polynomial time if $a=\P$, and~in exponential time if $a=\I$ or $a=\Os$.
\end{compactenum}
\end{theorem}

\begin{proof} %%[of Theorem~\ref{thrm:partialobs-reduction}]
First, we present the proof for $a \neq \I$, assuming that player~$2$ has complete observation.
A similar construction where player-1 instead of player-2 has complete observation is obtained symmetrically.
Let $G=\tuple{S_A \cup S_P, A_1, A_2, \trans, \Obs_1}$ and assume w.l.o.g. (according 
to Theorem~\ref{theo:separation} and Theorem~\ref{theo:uniform})
that $G$ satisfies interaction separation (i.e., states in $S_A$ are deterministic states, and $S_P$ are
probabilistic states) and $G$ is uniform-$n$-ary, i.e. all probabilities are equal to $\frac{1}{n}$.
For each probabilistic state $s \in S_P$, let $\Succ(s) = \tuple{s'_0,\dots,s'_{n-1}}$ be the $n$-tuple 
of states such that $\trans(s,-,-)(s'_i) = \frac{1}{n}$ for each $1 \leq i \leq n$.

%We apply the following transformation for each probabilistic state $s \in S$. 
We present a reduction that replaces the probabilistic states in $G$ by a gadget with player-1 and player-2 turn-based states.
From $G$, we construct the one-sided complete-observation game $G'$ where player-2 has complete observation.
%A similar construction where player-1 instead of player-2 has complete observation is obtained symmetrically.
The game $G' = \tuple{S',A_1',A_2',\trans',\Obs'_1}$ is defined as follows:
$S' = S \cup (S \times [n]) \cup \{\sink\}$, 
$A_1' = A_1 \cup [n]$,
$A_2' = A_2 \cup [n]$,
%$\Obs_1' = \{ o \cup \{(s,i) \mid s \in o \land i \in [n]\} \mid o \in \Obs_1\}$,
$\Obs_1' = \{ o \cup (o \times [n]) \mid o \in \Obs_1 \}$, 
and $\trans'$ is obtained from $\trans$ by applying the following transformation for each state $s \in S$:

\begin{enumerate}
\item if $s$ is a deterministic state in $G$, then $\trans'(s,a,b) = \trans(s,a,b)$ for all $a \in A_1, b \in A_2$, 
and $\trans'(s,i,j) = \sink$ for all $i,j \in [n]$;

\item if $s$ is a probabilistic state in $G$, then $s$ is a player-2 state in $G'$ 
and for all $i,j \in [n]$ we define $\trans'(s,-,i) = (s,i)$ and $\trans'((s,i),j,-) = s'_k$ 
such that $s'_k$ is the element in position $k$ in $\Succ(s)$ with $k= i+j \!\mod n$ 
(and let $\trans'(s,-,b) = \trans'((s,i),a,-) = \trans'(\sink,-,-) = \sink$ 
for all $a \in A_1, b \in A_2$). 

\end{enumerate}

Note that turn-based states in $G$ remain turn-based in $G'$ and the states 
$(s,i)$ are player-1 states with the same observation as~$s$.
As usual, the objective $\varphi'$ is defined as the set of plays in $G'$ whose
projection on $S^{\omega}$ belongs to $\varphi$.
%Given a play $\rho = s_0, s_1, (s_1, \cdot), s_2, s_3, (s_3,\cdot), s_4, \dots$
%in $G'$ (that does not visit $\sink$), 
%let $\mu(\rho) = s_0, s_1, s_2, s_3, s_4, \dots$ be the projection of $\rho$
%on $S^{\omega}$ that removes the intermediate states in $S \times [n]$.
%%A sequence of states $s_0, \ldots, s_m$ in $G$ corresponds to the sequence
%%$s_0, s_1, s_1, s_2, s_3, s_3, \ldots, s_m$ in $G'$ 
%%because deterministic and 
%%probabilistic states alternate in $G$, and in $G'$ transitions from probabilistic states have intermediate states. 
%The objective $\varphi'$ is defined as the set
%$\{\rho\in \Play(G') \mid \mu(\rho) \in \varphi\}$. 

Intuitively, each player in $G'$ has the possibility to ensure exact simulation of the probabilistic states 
of $G$ by playing actions in $[n]$ uniformly at random. For instance, if player~1 does so, then irrespective of the 
(possibly randomized) choice of player~$2$ among the states $(s,1), \dots, (s,n)$, the states in $\Succ(s)$
are reached with probability $1/n$, as in $G$. 
The same property holds if player~$2$ plays the actions in $[n]$ uniformly at random, 
no matter what player~$1$ does.
Therefore, by arguments similar to the proof of Theorem~\ref{theo1},
player~$1$ can ensure the objective $\varphi'$ in $G'$ is satisfied with the same probability 
as $\varphi$ in $G$, against any strategy of player~$2$, and the reduction is threshold-preserving.
% Given a strategy $\straa$ for
% player~$1$ in $G$, we define the strategy $\straa'$ in $G'$ as follows:
% for all $\rho \in \Pref(G)$, if $\obs(\Last(\rho)) \in \Obs$, then $\straa'(\rho) = \straa(\rho \setminus \{\tau\}$ 
% (where $\rho \setminus \{\tau\}$ is obtained from $\rho$ by projecting out the actions in $[n]$ and the 
% intermediate states with observation $\tau$), and if $\obs(\Last(\rho)) = \tau$, then $\straa'(\rho)$ is the
% uniform distribution over the set $[n]$.

The reduction can be easily adapted to the case $a = \I$ of games with 
partial information for both players. Since the construction of $G'$ is polynomial,
the complexity result (item (2) of the theorem) follows from Theorem~\ref{theo:separation} 
and Theorem~\ref{theo:uniform}.
\qed
\end{proof}

%1) explain that it is easy when all states of the probabilistic automaton have 
%the same number of outgoing transitions, each with the same probability.
%2) explain the encoding of probabilities into "concurrent game".

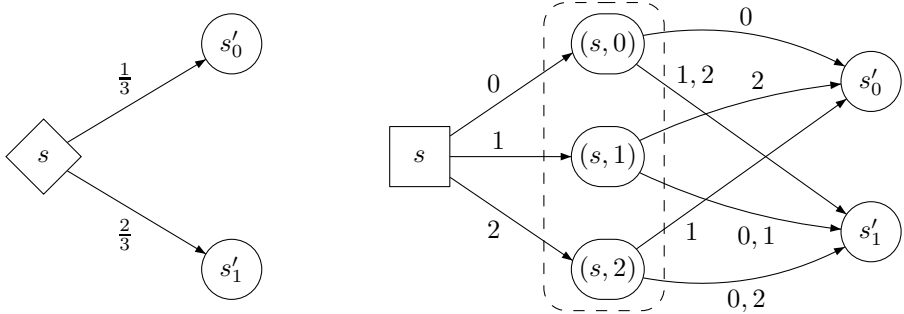
\begin{figure}[!tb]
  \begin{center}
    \hrule
    \begin{picture}(125,50)(0,0)
%\put(0,0){\framebox(125,50){}}

%\gasset{Nw=9,Nh=9,Nmr=4.5,rdist=1, loopdiam=6}

\rpnode[Nmarks=n](n0)(5,25)(4,5){$s$}
\node[Nmarks=n](n1)(30,40){$s'_0$}
\node[Nmarks=n](n2)(30,10){$s'_1$}

%\drawloop[ELside=l,loopCW=y](nk){$0,1$}

\drawedge[ELpos=50, ELside=l, curvedepth=0](n0,n1){$\frac{1}{3}$}
\drawedge[ELpos=50, ELside=r, curvedepth=0](n0,n2){$\frac{2}{3}$}

\node[Nmarks=n, Nmr=0](n0)(55,25){$s$}
\node[Nmarks=n, Nadjust=w](n1)(80,40){$(s,0)$}
\node[Nmarks=n, Nadjust=w](n2)(80,25){$(s,1)$}
\node[Nmarks=n, Nadjust=w](n3)(80,10){$(s,2)$}

\node[Nmarks=n, Nw=16, Nh=41, Nmr=3, dash={1.5}0, ExtNL=y, NLangle=270, NLdist=2](A1)(79.5,25){}

\node[Nmarks=n](m1)(115,35){$s'_0$}
\node[Nmarks=n](m2)(115,15){$s'_1$}

%\drawloop[ELside=l,loopCW=y](nk){$0,1$}

\drawedge[ELpos=45, ELside=l, eyo=2, curvedepth=0](n0,n1){$0$}
\drawedge[ELpos=42, ELside=l, curvedepth=0](n0,n2){$1$}
\drawedge[ELpos=45, ELside=r, eyo=-2, curvedepth=0](n0,n3){$2$}

\drawedge[ELpos=50, ELside=l, curvedepth=4](n1,m1){$0$}
\drawedge[ELpos=28, ELside=l, ELdist=0.5, curvedepth=0](n1,m2){$1,2$}
\drawedge[ELpos=60, ELside=l, curvedepth=2](n2,m1){$2$}
\drawedge[ELpos=60, ELside=r, curvedepth=-2](n2,m2){$0,1$}
\drawedge[ELpos=28, ELside=r, curvedepth=0](n3,m1){$1$}
\drawedge[ELpos=50, ELside=r, curvedepth=-4](n3,m2){$0,2$}

%\drawedge[dash={1}0](n3bis,nkbis){$0,1$}

\end{picture}
    \hrule
  \vspace{6pt}
    \caption{For the probabilistic state $s$ (on the left), we have $\Succ(s) = \tuple{s'_0,s'_1,s'_1}$
	and $n=3$ is the gcd of the probabilities denominators. Therefore, we apply the reduction of Theorem~\ref{thrm:partialobs-reduction}
	to obtain the turn-based game on the right, where $s$ is a player-2 state. \label{fig:gamify}}
  \end{center}
\end{figure}

%\subsection{Remarks and Consequences}\label{sec:remarks}
\subsection{Impossibility Results}\label{sec:remarks}

We have shown that for \PC\ games and \OT\ games, randomness is for free in the
transition function. 
We complete the picture (\figurename~\ref{figure-overview}) by showing that 
for \PT\ (complete-observation turn-based) games, randomness in the transition 
function cannot be obtained for free.

\begin{remark}[Role of probabilistic transition in \PT\ games and 
\POMDP s]\label{rem:tab1}
%We have shown that for \PC\ games and \OT\ games, randomness in 
%transition can be obtained for free. 
%We complete the picture by showing that for \PT\ (complete-observation  
%turn-based) games randomness in transition cannot be obtained for free. 
It follows from the result of Martin~\cite{Mar98} that for all \PT\ deterministic 
games and all objectives, the values are either~1 or~0; however, even
\MDP s with reachability objectives can have values in the interval $[0,1]$ 
(not value~0 and~1 only).
%Thus the result follows for \PT\ games. 
It follows that ``randomness in the transition function" cannot be replaced by 
``randomness in the strategies" in \PT\ deterministic games.
%in \PT\ deterministic games even with randomized 
%strategies the values are either~1 or~0~\cite{Mar98}; whereas \MDP s can have values in 
%the interval $[0,1]$.
For \POMDP s, we show in Theorem~\ref{thrm:pomdp} that pure strategies are 
sufficient, and it follows that for \POMDP s with deterministic transition 
function the values are~0 or~1, and since \MDP s with reachability objectives
can have values other than~0 and~1 it follows that randomness in the transition 
function cannot be obtained for free for \POMDP s. 
The probabilistic transitions also play an important role in the complexity of 
solving games in case of \PT\ games: for example, \PT\ deterministic 
games with reachability objectives can be solved in linear time, but with 
probabilistic transition function the problem is in NP~$\cap$~coNP and no 
polynomial-time algorithm is known.
In contrast, for \PC\ games we present a polynomial-time reduction from 
probabilistic to deterministic transition function. 
%%% for all parity objectives.
Table~\ref{tab:transitions} summarizes our results characterizing the 
classes of games where randomness in the transition function can be 
obtained for free.
\end{remark}

\newcommand{\ok}{free}
\newcommand{\ko}{not}

\begin{table}[t]
\begin{center}
\scalebox{.88}{
\begin{tabular}{|l|c|c|c|c|c|c|}
\cline{2-6}
\multicolumn{1}{l|}{{\large \strut}}   & \multicolumn{3}{|c|}{$2$\half-player} & \multicolumn{2}{|c|}{$1$\half-player} \\
\cline{2-6}
\multicolumn{1}{l|}{{\large \strut}}   & $\ $complete$\ $ & $\ $one-sided$\ $ & $\ $partial$\ $       & $\ $MDP$\ $ & $\ $POMDP$\ $  \\
\hline
$\ $turn-based$\ ${\large \strut}   & \ko\ (Rmk.~\ref{rem:tab1}) & \ok\ (Th.~\ref{thrm:partialobs-reduction}) & \ok\ (Th.~\ref{thrm:partialobs-reduction})                     & \ko\ (Rmk.~\ref{rem:tab1}) & \ko\ (Rmk.~\ref{rem:tab1})    \\
\hline
$\ $concurrent$\ ${\large \strut}   & \ok\ (Th.~\ref{theo1}) & \ok\ (Th.~\ref{theo1}) & \ok\ (Th.~\ref{theo1})                       & (NA) & (NA)  \\
\hline
\end{tabular}
}
\caption{When randomness is for free in the transition function. In particular, probabilities can be eliminated 
in all classes of 2-player games except complete-observation turn-based games. In the table,
Rmk.~\ref{rem:tab1} refers to Remark~\ref{rem:tab1}, 
Th.~\ref{thrm:partialobs-reduction} refers to Theorem~\ref{thrm:partialobs-reduction}, and 
Th.~\ref{theo1} refers to Theorem~\ref{theo1}. \label{tab:transitions}}
\end{center}
\end{table}

\subsection{Concurrency for free}\label{sec:concurrency4free}
The idea of the reduction in Theorem~\ref{thrm:partialobs-reduction}
can be extended to prove that concurrency is for free in one-sided complete-observation
games, i.e., 
%\PC\ games can be simulated by \OT\/ games)
%\smallskip\noindent{\bf Concurrency for free.} 
%We show that concurrency
%can be obtained for free in the presence of partial observation.
we present a polynomial reduction of \OC\ games to \OT\ games, 
and from \IC\ games to \IT\ games.

\begin{theorem}\label{thrm:concurrency4free}
There exists a reduction from \OC\/ games to \OT\/ games,
and from \IC\/ games to \IT\/ games, such that these reductions
are 
\begin{compactenum}
\item threshold-preserving, and
\item computable in polynomial time.
\end{compactenum}
\end{theorem}

\begin{proof} %%[of Theorem~\ref{thrm:partialobs-reduction}]
%Theorem~\ref{thrm:partialobs-reduction}
We present the reduction from \OC\/ games to \OT\/ games, for the case where 
player~$1$ has complete information. The reduction for one-sided games where
player~$2$ has complete information is symmetric. Finally, the reduction 
from \IC\/ games to \IT\/ games is obtained analogously.

Let $G=\tuple{S, A_1, A_2, \trans, \Obs_2}$ be a \OC\/ game where 
player~$1$ has complete information, and we construct
a \OT\/ game $G' = \tuple{S',A_1,A_2,\trans',\Obs'_1}$ as follows:
\begin{enumerate}
\item $S' = S \cup (S \times A_1)$, 
%\item $\Obs_2' = \{ o \cup \{(s,a) \mid s \in o \land a \in A_1\} \mid o \in \Obs_2\}$,
\item $\Obs_2' = \{ o \cup (o \times A_1) \mid o \in \Obs_2 \}$,
and 
\item $\trans'$ is defined as follows, for each state $s \in S$ and actions $a \in A_1$, $b\in A_2$:
$\trans'(s,a,-) = (s,a)$ and $\trans'((s,a),-,b) = \trans(s,a,b)$. 
\end{enumerate}

Hence the transition function $\trans'$ lets player~$1$ play first an action $a$, 
then player~$2$ plays an action $b$, and the successor state of $s$ is chosen 
according to the transition relation $\trans(s,a,b)$ from the original game.
As usual, the objective $\varphi'= \{s_0 (s_0,a_0) s_1 (s_1,a_1) \dots \mid 
s_0 s_1 \dots \in \varphi \land \forall i\geq0: a_i \in A_1\}$ in $G'$ 
requires that the projection of a play on $S^{\omega}$ satisfies $\varphi$.
Since player~$1$ plays first in $G'$, player~$1$ can achieve the objective $\varphi'$ 
in $G'$ with at most the same probability as for $\varphi$ in $G$, and since for all 
$s \in S$ and actions $a \in A_1$, the states~$s$ and $(s,a)$ are indistinguishable 
for player~$2$, player~$2$ does not know the last action chosen by player~$1$ and 
therefore does not gain any advantage in playing after player~$1$ rather than 
concurrently. 
Therefore the reduction is threshold-preserving and since it is computable in 
polynomial time, the result follows.
%Therefore, the value function for player~$1$ and objective $\varphi'$ in $G'$ is the 
%same as the same as for $\varphi$ in $G$, that is $\va^G(\varphi)(s) = \va^{G'}(\varphi')(s)$ for all $s \in S$.
\qed
\end{proof}

\smallskip\noindent{\bf Role of concurrency in complete-observation games.} 
We have shown that concurrency can be obtained for free in partial-observation games (\OT\ and \IT\ games).
In contrast, for complete-observation games, the value is irrational in general for
concurrent games with deterministic transitions (\PC\ deterministic games)~\cite{CH05b}, 
while the value is always rational in turn-based stochastic games with rational probabilities (\PT\ stochastic games)~\cite{CJH04}.
This rules out any value-preserving reduction of \PC\ (deterministic) games to \PT\ (stochastic) games with rational probabilities.

%Finally, note that it can be expected that randomness would not be for free in both the
%transition function and the strategies, and the results of this paper show that 
%the classes of games in which randomness is for free in the
%transition function (Table~\ref{tab:transitions}) are those in which randomized strategies are more powerful
%than pure strategies (Table~\ref{tab:strategies}), i.e. randomness is not for free 
%in strategies when randomness is for free in the transition function. 

%\vspace{-1em}
%\input{hugo}
\section{Randomness for Free in Strategies}\label{sec:strategies-pomdp}
\newcommand{\lef}[1]{L(#1)}
\newcommand{\rig}[1]{R(#1)}
\newcommand{\indic}[1]{\mathbf{1}_{#1}}
\newcommand{\esper}[1]{\mathbb{E}_s^\straa\left(#1\right)}
\newcommand{\maps}{f_{s_*,\straa}}
\newcommand{\ouv}[1]{\mathcal{C}(#1)}  %% old notation \mathcal{O} clashes with \Obs

In this section we present our results for randomness for free in strategies.
We start with a remark.

\begin{remark}[Randomness in strategies]\label{rem:tab2}
It is known from the results of~\cite{Eve57} that in \PC\  games 
randomized strategies are more powerful than pure strategies: 
values achieved by pure strategies are lower than values achieved by 
randomized strategies and randomized almost-sure  
winning strategies may exist whereas no pure almost-sure winning strategy 
exists.
Similar results also hold in the case of \OT\ games 
(see~\cite{CDHR07} for an example, also see Example~\ref{ex:example-one}).
By contrast we show that in POMDPs, restricting the set of strategies
to pure strategies does not decrease the value
nor affect the existence of almost-sure and positive winning strategies.
\end{remark}

We start with a lemma, %then present a result that can be derived from Martin's theorem for Blackwell games~\cite{Mar98}, 
and then present our results precisely in Theorem~\ref{thrm:pomdp}. 
The main argument in the proof of Lemma~\ref{lem:mdp} relies
on showing that the value $\Prb_{s}^{\straa}(\varphi)$ of any randomized 
observation-based strategy $\straa$ is equal to the average of the values 
$\Prb_{s}^{\straa_i}(\varphi)$ of (uncountably many)
pure observation-based strategies $\straa_i$. Therefore, one of the 
pure strategies $\straa_i$ has to achieve at least the value of the 
randomized strategy $\straa$.
The theory of integration and Fubini's theorem
make this argument precise. % (see~\cite{OurTechRpt} for details).

\begin{lemma}\label{lem:mdp}
Let $G$ be a \POMDP\/ (with countable state space~$S$), let $s_* \in S$ be an initial
state, and let $\varphi\subseteq S^\omega$ be an objective.
For every randomized observation-based strategy $\straa \in \Straa_G^O$
there exists a \emph{pure} observation-based strategy
$\straa_P \in \Straa_G^P \cap \Straa_G^O$
such that $\Prb_{s_*}^{\straa}(\varphi) \leq
\Prb_{s_*}^{\straa_P}(\varphi).$
\end{lemma}

\begin{proof}
Let $G=\tuple{S,A_1,\trans,\Obs_1}$ be a \POMDP\/ (remember that $A_2$ is a singleton in \POMDP{}s
and therefore $\Obs_2$ is irrelevant), 
let $\straa: \Pref(G) \to \distr(A_1)$ be a randomized observation-based strategy,
and fix $s_*\in S$ an initial state.

To simplify notations,
we suppose that $A_1= \{0,1\}$ contains only two actions,
and that given a state $s\in S$ and an action $a\in\{0,1\}$
there are only two possible successors $\lef{s,a}\in S$
and $\rig{s,a}\in S$ chosen with respective probabilities
$\trans(s,a,\lef{s,a})$ and
$\trans(s,a,\rig{s,a}) = 1 -\trans(s,a,\lef{s,a})$.
The proof for an arbitrary finite set of actions and more than two successors
is essentially the same, with more complicated notations.

There is a natural way to ``derandomize'' the randomized strategy $\straa$.
Fix an infinite sequence $x=(x_n)_{n\in\nat}\in[0,1]^\omega$
and define the pure strategy $\straa_x: \Pref(G) \to A_1$ as follows.
For every play prefix $h = s_0 \,a_1\, s_1 \,a_2\, s_2 \ldots s_n$,
\[
\straa_x(h)=
\begin{cases}
0 & \text{ if } x_n \leq \straa(h)(0)\\
1 & \text{ otherwise.}
\end{cases}
\]
Intuitively, the sequence $x$ fixes in advance the sequence of results of coin
tosses used for playing with $\straa$. 
Note that if $\sigma$ is observation-based, then for every sequence $x$ the strategy $\sigma_x$ 
is both observation-based and pure.

To prove the lemma, we show that
$[0,1]^\omega$ can be equipped with a probability measure $\nu$
such that the mapping $x\mapsto
\Prb_{s_*}^{\straa_x}(\varphi)$ from $[0,1]^\omega$ to $[0,1]$ is measurable,
and:
\begin{equation}\label{eq:mainn}
\Prb_{s_*}^{\straa}(\varphi)=
\int_{x\in[0,1]^\omega}
\Prb_{s_*}^{\straa_x}(\varphi)~  d\nu(x)\enspace.
\end{equation}

Suppose that~\eqref{eq:mainn} holds.
Then there exists $x\in[0,1]^\omega$
(actually many $x$'s) such that $\Prb_{s_*}^{\straa}(\varphi) \leq
\Prb_{s_*}^{\straa_x}(\varphi)$
and since strategy $\straa_x$ is deterministic,
this proves the lemma.

To complete the proof, it is thus enough to construct a probability measure
$\nu$ on $[0,1]^\omega$ such that~\eqref{eq:mainn} holds.

We start with the definition of the probability measure $\nu$.
The set $[0,1]^\omega$ is equipped with the sigma-field generated by
\emph{sequence-cylinders} which are defined as follows.
For every finite sequence $x=x_0,x_1,\ldots,x_n\in[0,1]^*$
the sequence-cylinder $\ouv{x}$
is the subset
$[0,x_0]\times[0,x_1] \times\ldots\times [0,x_n]\times
[0,1]^\omega\subseteq [0,1]^\omega$.
According to Tulcea's theorem~\cite{Billingsley},
there is a unique product probability measure
$\nu$ on $[0,1]^\omega$ such that $\nu(\ouv{\epsilon})=1$ and for every sequence
$x_0,\ldots,x_n,x_{n+1}$ in $[0,1]$,
\[
\nu(\ouv{x_0,\ldots,x_n,x_{n+1}})=x_{n+1}\cdot
\nu(\ouv{x_0,\ldots,x_n})\enspace.
\]

Now that $\nu$ is defined, it remains to prove that the mapping $x\mapsto
\Prb_{s_*}^{\straa_x}(\varphi)$ from $[0,1]^\omega$ to $[0,1]$ is measurable
and that~\eqref{eq:mainn} holds.
For that, we introduce the following mapping:
\[
\maps : [0,1]^\omega \times [0,1]^\omega \to (SA_1)^\omega,
\]
that associates with every pair of sequences
$((x_n)_{n\in\nats},(y_n)_{n\in\nats})$
the infinite history
$h=s_0 \,a_1\, s_1 \,a_2\, \ldots \in(S A_1)^\omega$
defined recursively as follows.
First $s_0=s_*$, and for every $n\in\nats$,
\[
a_{n+1} =
\begin{cases}
0 & \text{if } x_n\leq \straa(s_0 \,a_1\, s_1 \cdots s_n)(0),\\
1 & \text{otherwise.}
\end{cases}
\]
\[
s_{n+1} =
\begin{cases}
\lef{s_n,a_{n+1}} & \text{if } y_n\leq \trans(s_n,a_{n+1},\lef{s_n,a_{n+1})},\\
\rig{s_n,a_{n+1}} & \text{otherwise.}
\end{cases}
\]

Intuitively, $(x_n)_{n\in\nats}$ fixes in advance the coin tosses used by
the strategy, while $(y_n)_{n\in\nats}$ takes care of the coin tosses used by the
probabilistic transitions, and $\maps$ produces the resulting description of
the play.
Thanks to the mapping $\maps$, randomness related to the use of the
randomized strategy $\straa$ is separated from randomness
due to transitions of the game,
which allows to represent the randomized strategy $\straa$
by mean of a probability measure over the set of pure strategies
$\{\straa_x \mid x\in[0,1]^\omega\}$.

%%\medskip

We equip both sets $(SA_1)^\omega$ and $[0,1]^\omega\times[0,1]^\omega$
with sigma-fields that make $\maps$ measurable.
First,
$(SA_1)^\omega$
is equipped with the sigma-field
generated by cylinders, defined as follows.
An \emph{action-cylinder} is a subset
$\ouv{h}\subseteq(SA_1)^\omega$ such that $\ouv{h}=h(SA_1)^\omega$ for some $h\in(SA_1)^*$.
A \emph{state-cylinder} is a subset
$\ouv{h}\subseteq(SA_1)^\omega$
such that $\ouv{h}=h(A_1S)^\omega$ for some $h\in(SA_1)^*S$.
The set of \emph{cylinders} is the union of the sets of action-cylinders and
state-cylinders.
Second,
$[0,1]^\omega\times[0,1]^\omega$ is equipped with the sigma-field
generated by products of sequence-cylinders.
Checking that $\maps$ is measurable is an elementary exercise.

\medskip

Now we define two probability measures $\mu$ and $\mu'$ on $(SA_1)^\omega$
and prove that they coincide.
On one hand, the measurable mapping $\maps:
[0,1]^\omega\times[0,1]^\omega\to(SA_1)^\omega$
defines naturally a probability measure $\mu'$ on $(SA_1)^\omega$.
Equip the set $[0,1]^\omega\times[0,1]^\omega$ with the
product measure $\nu\times \nu$.
Then for every measurable subset $B \subseteq(SA_1)^\omega$,
\[
\mu'(B) = (\nu\times\nu)(\maps^{-1}(B))\enspace.
\]
On the other hand, the strategy $\straa$ and the initial state $s_*$
naturally define another probability measure $\mu$ on $(SA_1)^\omega$.
According to Tulcea's theorem~\cite{Billingsley},
there exists a unique product probability measure $\mu$ on $(SA_1)^\omega$ such
that $\mu(\ouv{s_*})=1$, $\mu(\ouv{s})=0$ for $s\in S \setminus\{s_*\}$, 
and for $h=s_0 \,a_1\, s_1 \,a_2\, \cdots s_n\in(SA_1)^*S$ and $(a,t)\in A_1\times S$,
  \begin{align*}
	\mu(\ouv{ha}) &= \mu(\ouv{h}) \cdot \straa(h)(a) \\
	\mu(\ouv{hat}) &= \mu(\ouv{ha}) \cdot \trans(s_n,a,t).
  \end{align*}

%We have defined $\maps$ in such a way that $\mu$ and $\mu'$ coincide.
To prove that $\mu$ and $\mu'$ coincide,
it is enough to prove that $\mu$ and $\mu'$ coincide on
the set of cylinders,
%a set of subsets
%of $(SA)^\omega$ which generates the
%$\sigma$-field,
%A good choice for $\mathcal{C}$ is the set of sets
%of the form $F_{n,s}=(SA)^ns(AS)^\omega$ for some $n\in\nats$ and $s\in S$
%together with sets $F_{n,a}=(SA)^nSa(SA)^\omega$
%for some $n\in\nats$ and $a\in A$.
that is
%Now we have to prove that for every $f\in S\cup A$,
for
every cylinder $\ouv{h}\subseteq (SA_1)^\omega$,
\begin{equation}\label{eq:egaa}
\mu(\ouv{h}) = (\nu\times\nu)(\maps^{-1}(\ouv{h}))\enspace.
\end{equation}
This is obvious for $h=s_*$ and $h=s\in S\setminus\{s_*\}$. 
The general case goes by induction.
Let $h=s_0 \,a_1\, s_1 \,a_2\, \cdots s_n\in(SA_1)^*S$ and $(a,t)\in A_1\times S$.
Let $I=[0,1]$. Let $I_a = [0,\straa(h)(a)]$ if $a=0$ and $I_a=[\straa(h)(a),1]$
if $a=1$. Let $I_t=[0,\trans(s_n,a,t)]$ if $t=\lef{s_n,a}$
and $I_t=[\trans(s_n,a,t),1]$ if $t=\rig{s_n,a}$.
Then:
\begin{align*}
\mu(\ouv{ha} \mid \ouv{h}) &= \straa(h)(a)\\
& =
(\nu\times\nu)((I\times I)^n(I_a\times I)(I\times I)^\omega)
\\
&=(\nu\times\nu)(\maps^{-1}(\ouv{ha})\mid\maps^{-1}(\ouv{h}))\\
%
%(\nu\times\nu)
%(\ouv{x,x_{n+1}}\times\ouv{y}\mid
%\ouv{x}\times\ouv{y})\\
\mu(\ouv{hat} \mid \ouv{ha}) &= \trans(s_n,a,t)\\
& =
(\nu\times\nu)((I\times I)^n(I\times I_t)(I\times
I)^\omega)
\\
&=(\nu\times\nu)(\maps^{-1}(\ouv{hat})\mid\maps^{-1}(\ouv{ha}))\enspace,
\end{align*}
which proves that~\eqref{eq:egaa} holds for every cylinder $\ouv{h}$.

\medskip

Now all the tools needed to prove~\eqref{eq:mainn}
have been introduced, and we can state the main relation
between $\maps$ and $\Prb_{s_*}^{\straa}(\varphi)$.
Let
$\varphi'\subseteq(SA_1)^\omega$ be the set of histories $s_0 \,a_1\, s_1 \,a_2\, \ldots$ 
such that $s_0s_1\cdots \in \varphi$, and let $\indic{\varphi}$ and $\indic{\varphi'}$
be the indicator functions of $\varphi$ and $\varphi'$.
Then:
\begin{align}
\notag
\Prb_{s_*}^{\straa}(\varphi)
&=\int_{p\in S^\omega} \indic{\varphi}(p) ~d\Prb_{s_*}^{\straa}(p)
%%\\ \notag &
=\int_{p\in(SA_1)^\omega} \indic{\varphi'}(p) ~d\mu(p)
%%\\ \notag &
= \int_{p\in(SA_1)^\omega} \indic{\varphi'}(p) ~d\mu'(p)\\
\notag
&= \int_{(x,y)\in [0,1]^\omega\times [0,1]^\omega}
\indic{\varphi'}(\maps(x,y)) ~d(\nu\times\nu)(x,y)\\
&=
\int_{x\in[0,1]^\omega} \left( \int_{y\in[0,1]^\omega} \indic{\varphi'}(\maps(x,y))
~d\nu(y)\right)~d\nu(x)\enspace,
\label{eq:last}
\end{align}
where
the first and second equalities are by definition of
$\Prb_{s_*}^{\straa}(\varphi)$,
the third equality holds because $\mu=\mu'$,
the fourth equality is a basic property of image
measures, and the last equality holds by Fubini's theorem~\cite{Billingsley}
that we can use since $\indic{\varphi'}\circ\maps$ is positive.
\smallskip

\noindent To complete the proof, we show that for every $x\in[0,1]^\omega$,
\begin{equation}\label{eq:sigmax}
\int_{y\in [0,1]^\omega}
\indic{\varphi'}(\maps(x,y)) ~d\nu(y) = \Prb_{s}^{\straa_x}(\varphi),
\end{equation}
Equation~\eqref{eq:last} holds for every observation-based strategy
$\straa$, hence in particular for strategy $\straa_x$.
But strategy $\straa_x$ has the following property:
for every $x'\in\,]0,1[^\omega$ and every $y\in [0,1]^\omega$,
$f_{s_*,\straa_x}(x',y)=f_{s_*,\straa}(x,y)$.
Together with~\eqref{eq:last}, this gives~\eqref{eq:sigmax}.
This completes the proof, since~\eqref{eq:last} and~\eqref{eq:sigmax}
immediately give~\eqref{eq:mainn}.
\qed
\end{proof}

We obtain the following result as a consequence of Lemma~\ref{lem:mdp}.

\begin{theorem}\label{thrm:pomdp}
Let $G$ be a \POMDP\/ (with countable state space~$S$), let $s_* \in S$ be an 
initial state, and let $\varphi\subseteq S^\omega$ be an objective.
Then the following assertions hold:
\begin{enumerate}
\item
%%\begin{equation}\label{eq:final}
$\sup_{\straa \in \Straa_G^O} \Prb_{s_*}^{\straa}(\varphi) =  
\sup_{\straa \in \Straa_G^O \cap \Straa_G^P}  \Prb_{s_*}^{\straa}(\varphi)$. 
%%\enspace.
%%\end{equation}
\item If there is a randomized optimal (resp., almost-sure winning, positive winning) strategy 
for $\varphi$ from $s_*$, then there is a pure optimal (resp., almost-sure winning, positive winning)
strategy for $\varphi$ from $s_*$. 
\end{enumerate}
\end{theorem}

Theorem~\ref{thrm:pomdp} shows that the result of Theorem~\ref{thrm:martin} can be generalized to 
\POMDP s, and a stronger result (item (2) of Theorem~\ref{thrm:pomdp})  
can be proved for \POMDP s (and \MDP s as a special case).
It remains open whether a result similar to item (2) of 
Theorem~\ref{thrm:pomdp} can be proved for \PT\ stochastic games.
Note that it was already shown in~\cite[Example~1]{CMJ04} that in \PT\ stochastic
games with Borel objectives optimal strategies need not exist. 
The results summarizing when randomness can be obtained for free for 
strategies is shown in Table~\ref{tab:strategies}.

\begin{table}[t]
\begin{center}
\scalebox{.88}{
\begin{tabular}{|l|c|c|c|c|c|c|}
\cline{2-6}
\multicolumn{1}{l|}{{\large \strut}}   & \multicolumn{3}{|c|}{$2$\half-player} & \multicolumn{2}{|c|}{$1$\half-player} \\
\cline{2-6}
\multicolumn{1}{l|}{{\large \strut}}   & $\ $complete$\ $ & $\ $one-sided$\ $ & $\ $partial$\ $       & $\ $MDP$\ $ & $\ $POMDP$\ $  \\
\hline
$\ $turn-based$\ ${\large \strut}   & $\epsilon > 0$ (Th.~\ref{thrm:martin}) & \ko\ (Rmk.~\ref{rem:tab2}) & \ko\ (Rmk.~\ref{rem:tab2})              & 
$\epsilon \geq 0$ (Th.~\ref{thrm:pomdp}) & $\epsilon \geq 0$  (Th.~\ref{thrm:pomdp})   \\
\hline
$\ $concurrent$\ ${\large \strut}   & \ko\ (Rmk.~\ref{rem:tab2}) & \ko\ (Rmk.~\ref{rem:tab2}) & \ko\ (Rmk.~\ref{rem:tab2})    & (NA) & (NA)  \\
\hline
\end{tabular}
}
\caption{When pure ($\epsilon$-optimal) strategies are as powerful as randomized strategies. The case $\epsilon = 0$
in complete-observation turn-based games is open.
In the table, Th.~\ref{thrm:martin} refers to Theorem~\ref{thrm:martin},
Rmk.~\ref{rem:tab2} refers to Remark~\ref{rem:tab2}, 
Th.~\ref{thrm:pomdp} refers to Theorem~\ref{thrm:pomdp}.\label{tab:strategies}}
\end{center}
\end{table}

\smallskip\noindent{\bf Undecidability result for \POMDP s.}
The results of~\cite{BBG08} show that the emptiness problem for finite-state 
probabilistic coB\"uchi (resp., B\"uchi) automata under 
the almost-sure (resp., positive) semantics~\cite{BBG08} is undecidable.
As a consequence it follows that for finite-state \POMDP s the problem of deciding if there 
is a pure observation-based almost-sure (resp., positive) winning strategy 
for coB\"uchi (resp., B\"uchi) objectives is undecidable, and as a consequence 
of Theorem~\ref{thrm:pomdp} we obtain an analogous undecidability result for 
randomized strategies. 
%%This result closes an open question discussed in~\cite{GriponS09}.
The undecidability result holds even if the coB\"uchi (resp., B\"uchi) 
objectives is visible. 

\begin{corollary}\label{coro1}
Let $G$ be a finite-state  \POMDP\/ with initial state $s_*$
and let $\target\subseteq S$ be a subset of states (or union of observations).
Whether there exists a pure or randomized almost-sure winning strategy for player~1 
from $s_*$ in~$G$ for the objective~$\coBuchi(\target)$ is undecidable;
and whether there exists a pure or randomized positive winning strategy for 
player~1 from $s_*$ in~$G$ for the objective~$\Buchi(\target)$ is undecidable.
\end{corollary}

%%\smallskip
\noindent{\bf Undecidability result for one-sided complete-observation turn-based games.}
The undecidability results of Corollary~\ref{coro1} also holds for finite-state 
\OT\ stochastic games (as they subsume finite-state \POMDP s as a special case).
It follows from Theorem~\ref{thrm:partialobs-reduction} that finite-state \OT\ 
stochastic games can be reduced to finite-state \OT\ deterministic games.
The reduction holds for randomized strategies and thus we obtain the first undecidability 
result for finite-state \OT\ deterministic games (Corollary~\ref{coro2}), solving 
the open question of~\cite{CDHR07}. Note that for pure strategies, \OT\ deterministic 
games with a parity objective are EXPTIME-complete~\cite{Reif84,CDHR07}.

\begin{corollary}\label{coro2}
Let $G$ be a finite-state \OT\ deterministic game 
with initial state $s_*$ and let $\target\subseteq S$ be a subset of states 
(or union of observations).
Whether there exists a randomized almost-sure winning strategy for player~1 from $s_*$ 
in~$G$ for the objective~$\coBuchi(\target)$ is undecidable;
and whether there exists a randomized positive winning strategy for player~1 
from $s_*$ in~$G$ for the objective~$\Buchi(\target)$ is undecidable.
\end{corollary}

\section{Conclusion}
In this work we have presented a precise characterization for classes of
games where randomization can be obtained for free in transition functions
and in strategies.
As a consequence of our characterization we obtain new undecidability 
results. 
The other impact of our characterization is as follows: for the class of
games where randomization is free in transition function, future algorithmic 
and complexity analysis can focus on the simpler class of deterministic 
games; and for the class of games where randomization is free in 
strategies, future analysis of such games can focus on the simpler class
of pure strategies.
Thus our results will be useful tools for simpler analysis techniques 
in the study of games, as already demonstrated in~\cite{CCHRS11,CC14,CCT13,CD14,GO14,GS14}.

Finally, note that it can be expected that randomness would not be for free in both the
transition function and the strategies, and the results of this paper show that 
%randomness is for free in either the transitions or the strategies.
the classes of games in which randomness is for free in the
transition function (Table~\ref{tab:transitions}) are those in which randomized strategies are more powerful
than pure strategies (Table~\ref{tab:strategies}), i.e. randomness is not for free 
in strategies when randomness is for free in the transition function.

\bibliographystyle{plain}
\bibliography{diss}

\begin{thebibliography}{10}

\bibitem{AHK02}
R.~Alur, T.~A. Henzinger, and O.~Kupferman.
\newblock Alternating-time temporal logic.
\newblock {\em Journal of the ACM}, 49:672--713, 2002.

\bibitem{BBG08}
C.~Baier, N.~Bertrand, and M.~Gr{\"o}{\ss}er.
\newblock On decision problems for probabilistic {B}{\"u}chi automata.
\newblock In {\em FoSSaCS}, LNCS 4962, pages 287--301. Springer, 2008.

\bibitem{DBLP:conf/lics/BertrandGG09}
N.~Bertrand, B.~Genest, and H.~Gimbert.
\newblock Qualitative determinacy and decidability of stochastic games with
  signals.
\newblock In {\em Proc. of LICS}, pages 319--328. IEEE Computer Society, 2009.

\bibitem{Billingsley}
P.~Billingsley.
\newblock {\em Probability and {M}easure}.
\newblock Wiley-Interscience, 1995.

\bibitem{BuchiLandweber69}
J.~R. B{\"u}chi and L.~H. Landweber.
\newblock Solving sequential conditions by finite-state strategies.
\newblock {\em Transactions of the AMS}, 138:295--311, 1969.

\bibitem{CCHRS11}
P.~Cern{\'{y}}, K.~Chatterjee, T.~A. Henzinger, A.~Radhakrishna, and R.~Singh.
\newblock Quantitative synthesis for concurrent programs.
\newblock In {\em CAV}, pages 243--259, 2011.

\bibitem{CC14}
K.~Chatterjee and M.~Chmelik.
\newblock {POMDPs} under probabilistic semantics.
\newblock {\em CoRR}, abs/1408.2058, 2014 (Conference version: UAI, 2013).

\bibitem{CCT13}
K.~Chatterjee, M.~Chmelik, and M.~Tracol.
\newblock What is decidable about partially observable {Markov} decision
  processes with omega-regular objectives.
\newblock In {\em {CSL}}, pages 165--180, 2013.

\bibitem{CD14}
K.~Chatterjee and L.~Doyen.
\newblock Partial-observation stochastic games: How to win when belief fails.
\newblock {\em {ACM} Trans. Comput. Log.}, 15(2):16, 2014.

\bibitem{CDHR07}
K.~Chatterjee, L.~Doyen, T.~A. Henzinger, and J.-F. Raskin.
\newblock Algorithms for omega-regular games of incomplete information.
\newblock {\em Logical Methods in Computer Science}, 3(3:4), 2007.

\bibitem{CH05b}
K.~Chatterjee and T.~A. Henzinger.
\newblock Semiperfect-information games.
\newblock In {\em FSTTCS'05}. LNCS 3821, Springer, 2005.

\bibitem{CJH04}
K.~Chatterjee, M.~Jurdzi{\'{n}}ski, and T.~A. Henzinger.
\newblock Quantitative stochastic parity games.
\newblock In {\em SODA'04}, pages 121--130. SIAM, 2004.

\bibitem{CMJ04}
K.~Chatterjee, R.~Majumdar, and M.~Jurdzi{\'{n}}ski.
\newblock On {N}ash equilibria in stochastic games.
\newblock In {\em CSL'04}, pages 26--40. LNCS 3210, Springer, 2004.

\bibitem{InterfaceTheories}
L.~de~Alfaro and T.~A. Henzinger.
\newblock Interface theories for component-based design.
\newblock In {\em EMSOFT'01}, LNCS 2211, pages 148--165. Springer, 2001.

\bibitem{Eve57}
H.~Everett.
\newblock Recursive games.
\newblock In {\em Contributions to the Theory of Games {III}}, volume~39 of
  {\em Annals of Mathematical Studies}, pages 47--78, 1957.

\bibitem{GO14}
H.~Gimbert and Y.~Oualhadj.
\newblock Deciding the value 1 problem for {$\sharp$}-acyclic partially
  observable {M}arkov decision processes.
\newblock In {\em {SOFSEM}}, pages 281--292, 2014.

\bibitem{GS14}
J.~Goubault{-}Larrecq and R.~Segala.
\newblock Random measurable selections.
\newblock In {\em Horizons of the Mind}, pages 343--362, 2014.

\bibitem{FairSimulation}
T.~A. Henzinger, O.~Kupferman, and S.~Rajamani.
\newblock Fair simulation.
\newblock {\em Information and Computation}, 173:64--81, 2002.

\bibitem{Kechris}
A.~Kechris.
\newblock {\em Classical Descriptive Set Theory}.
\newblock Springer, 1995.

\bibitem{Mar98}
D.~A. Martin.
\newblock The determinacy of {Blackwell} games.
\newblock {\em The Journal of Symbolic Logic}, 63(4):1565--1581, 1998.

\bibitem{McNaughton93}
R.~McNaughton.
\newblock Infinite games played on finite graphs.
\newblock {\em Annals of Pure and Applied Logic}, 65:149--184, 1993.

\bibitem{MSZ94}
J.-F. Mertens, S.~Sorin, and S.~Zamir.
\newblock Repeated games.
\newblock {\em Core Discussion Papers}, 9422, 1994.

\bibitem{PnueliRosner89}
A.~Pnueli and R.~Rosner.
\newblock On the synthesis of a reactive module.
\newblock In {\em POPL'89}, pages 179--190. ACM Press, 1989.

\bibitem{RamadgeWonham87}
P.~J. Ramadge and W.~M. Wonham.
\newblock Supervisory control of a class of discrete-event processes.
\newblock {\em SIAM Journal of Control and Optimization}, 25(1):206--230, 1987.

\bibitem{Reif84}
J.~H. Reif.
\newblock The complexity of two-player games of incomplete information.
\newblock {\em Journal of Computer and System Sciences}, 29(2):274--301, 1984.

\bibitem{Thomas97}
W.~Thomas.
\newblock Languages, automata, and logic.
\newblock In {\em Handbook of Formal Languages}, volume 3, Beyond Words,
  chapter~7, pages 389--455. Springer, 1997.

\bibitem{Var85}
M.~Y. Vardi.
\newblock Automatic verification of probabilistic concurrent finite-state
  systems.
\newblock In {\em FOCS}, pages 327--338. IEEE Computer Society Press, 1985.

\bibitem{ZP95}
U.~Zwick and M.~Paterson.
\newblock The complexity of mean payoff games on graphs.
\newblock {\em Theoretical Computer Science}, 158:343--359, 1996.

\end{thebibliography}

\end{document}